%% file: document_R.tex
\documentclass[journal,draftcls,onecolumn,12pt,twoside]{IEEEtran}
\usepackage[T1]{fontenc}
\usepackage{cite}
\pdfoutput=1

\usepackage[dvipsnames]{xcolor}

\ifCLASSINFOpdf
\else
\fi
\usepackage{amsthm}
\usepackage{amsmath}
\usepackage{amsfonts}
\usepackage[cmintegrals]{newtxmath}

\usepackage{array}

\usepackage{stfloats}

\usepackage{bbm}
\usepackage[pdftex]{graphicx}

\usepackage{caption}
\usepackage{siunitx}
\usepackage{tikz}
\usepackage{pgfplots}
\usepackage{verbatim}
\usepackage{color,soul}
\usetikzlibrary{spy,calc}
\newtheorem{theorem}{Theorem}
\newtheorem{lemma}{Lemma}

\newtheorem{remark}{Remark}

\usetikzlibrary{calc,positioning}
\usetikzlibrary{plotmarks}
\usepackage{multirow,bigdelim,dcolumn,booktabs}
\usepackage{hhline}
\usepackage{colortbl}

\begin{document}

\title{The Exact Rate Memory Tradeoff for Small Caches with Coded Placement}

\author{Vijith Kumar K P,~
	Brijesh Kumar Rai~and~Tony Jacob}

\maketitle

\begin{abstract}
	The idea of coded caching was introduced by Maddah-Ali and Niesen who demonstrated the advantages of coding in caching problems. To capture the essence of the problem, they introduced the $(N, K)$ canonical cache network in which $K$ users with independent caches of size $M$ request files from a server that has $N$ files. Among other results, the caching scheme and lower bounds proposed by them led to a characterization of the exact rate memory tradeoff when $M\geq \frac{N}{K}(K-1)$. These lower bounds along with the caching scheme proposed by Chen et al.  led to a characterization of the exact rate memory tradeoff when $M\leq \frac{1}{K}$. In this paper we focus on small caches where $M\in \left[0,\frac{N}{K}\right]$ and derive new lower bounds. For the case when $\big\lceil\frac{K+1}{2}\big\rceil\leq N \leq K$ and $M\in \big[\frac{1}{K},\frac{N}{K(N-1)}\big]$, our lower bounds demonstrate that the caching scheme introduced by G{\'o}mez-Vilardeb{\'o} is optimal and thus extend the characterization of the exact rate memory tradeoff. For the case $1\leq N\leq \big\lceil\frac{K+1}{2}\big\rceil$, we show that the new lower bounds improve upon the previously known lower bounds.
	
\end{abstract}

\IEEEpeerreviewmaketitle

\section{Introduction}

Content distribution networks use memory distributed across the network, known as caches, to reduce the peak time data traffic by keeping copies of file fragments near the end-users. These techniques, known as caching techniques, generally operate in two phases. In the first phase, called the placement phase, the server fills the caches with fragments of files available in the server. In the second phase, called the delivery phase, the server broadcasts a set of packets to meet each user's requests, aided by the caches available near to the user. Maddah-Ali and Niesen, in their seminal work \cite{maddah2014fundamental}, noted that traditional caching techniques fail to exploit the multicast opportunity available in such networks. To address this limitation, they introduced the notion of coded caching and proposed a scheme to demonstrate that coding reduces the peak data traffic load over traditional uncoded caching schemes. They introduced the canonical $(N, K)$ cache network where the server has $N$ files $\{W_{1},\dots,W_{N}\}$ and is communicating with $K$ users $\{U_{1},\dots, U_{K}\}$ through a common shared error-free link. Here, each user $U_{k}$ is equipped with an isolated cache $Z_{k}$ of size $M\in [0, N]$ as shown in Fig. \ref{Fig: (N,K) cache network}.
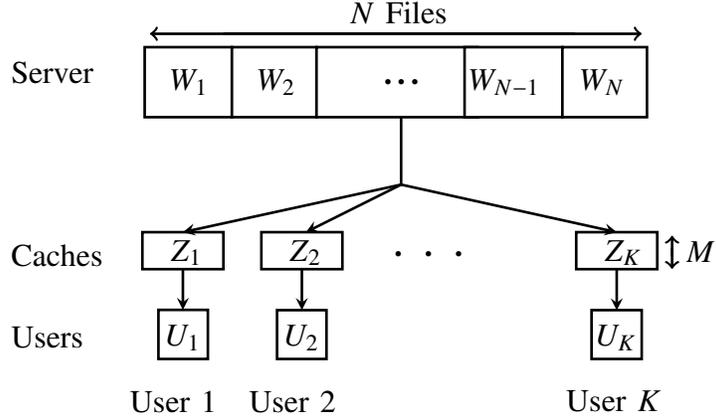
\begin{figure}[h]
	\centering
	\input{can}
	\caption{The $(N,K)$ cache network}
	\label{Fig: (N,K) cache network}
\end{figure}
During the placement phase, the server fills each user's cache without knowing their future demands. Let the users' requests be represented by $\textbf{\textit{d}}=(W_{d_{1}},\dots, W_{d_{K}})$, where $W_{d_{l}}$ is the file requested by user $U_l$. During the delivery phase, when demand $\textbf{\textit{d}}$ is revealed, the server broadcasts a set of packets $X_{\textbf{\textit{d}}}$ of size $R_{\textbf{\textit{d}}}(M)$ so that users can obtain their requested file, aided by their cache contents. The quantity $R_{\textbf{\textit{d}}}(M)$ is called the load experienced by the network. The fundamental issue in a caching scheme is to decide what to store in each user's cache and, accordingly, what to broadcast to fulfill each user's demand such that the load experienced by the shared link is minimized. Variations of this problem have been introduced for decentralized networks \cite{maddah2015decentralized}, hierarchical networks \cite{karamchandani2016hierarchical}, multiple servers \cite{shariatpanahi2016multi}, heterogeneous networks \cite{daniel2019optimization}, D2D networks \cite{yapar2019optimality}, and shared cache networks \cite{parrinello2019coded}. Issues of privacy \cite{ravindrakumar2017private} and data shuffling\cite{wan2018fundamental} have also been studied in this setup.

As a result of the inherent symmetry of the problem of coded caching for the canonical $(N,K)$ cache network, all demands related to each other through a permutation are naturally grouped together. This leads to the consideration of symmetric caching schemes which were shown by Tian \cite{tian2016symmetry} to operate at the same rate. Thus, we consider only the class of symmetric caching schemes in this paper. Consider a demand $\textbf{\textit{d}}$, where the user $U_l$ requires the file $W_{d_{l}}$,
\begin{equation}
\textbf{\textit{d}}=(W_{d_{1}},\dots, W_{d_{K}})
\end{equation}
Let $\pi(.)$ be a permutation operation defined over the set $\{1,\dots, K\}$ and $\pi^{-1}(.)$ be its inverse. Now consider another demand $\pi \textbf{\textit{d}}$, which is obtained by permuting the files requested by the users,
\begin{equation}
\pi \textbf{\textit{d}}=(W_{d_{\pi^{-1}(1)}},\dots,W_{d_{\pi^{-1}(K)}}).
\end{equation}
In the demand $\pi \textbf{\textit{d}}$, the user $U_{\pi(l)}$ requires the file $W_{d_{l}}$. In response to the demand $\pi \textbf{\textit{d}}$, the server broadcasts a set of packets $X_{\pi\textbf{\textit{d}}}$. For a symmetric caching scheme, we have \cite{tian2016symmetry}
\begin{equation}
\label{symmetry}
H(W_{d_{l}},Z_{\pi(l)},X_{\pi\textbf{\textit{d}}})=H(W_{d_{l}},Z_{l},X_{\textbf{\textit{d}}})
\end{equation}
Consider the demands where each of the $N$ files is required by at least one user (and hence $N\leq K$). The set of all such demands is denoted by $\textbf{\textit{D}}$ and the corresponding rate is denoted by $R(M)$, where 
\begin{equation}
R(M)=\max\{R_{\textbf{\textit{d}}}(M)\mid \textbf{\textit{d}}\in \textbf{\textit{D}}\}.
\end{equation}
For the $(N,K)$ cache network with cache size $M$, the memory rate pair $(M,R)$ is said to be achievable if there is a scheme with $R(M)\leq R$. For a such a  scheme, we have
\begin{IEEEeqnarray}{rl}
	\label{M}
	H(Z_{l})\leq& M\\
	\label{R}
	H(X_{\textbf{\textit{d}}})\leq& R\\
	\label{I.1}
	H(Z_{l},X_{\textbf{\textit{d}}})=&H(W_{d_{l}},Z_{l},X_{\textbf{\textit{d}}}),\\
	\label{I.2}
	H(W_{1},\dots,W_{N},Z_{l},X_{\textbf{\textit{d}}})=&H(W_{1},\dots,W_{N}),
\end{IEEEeqnarray}
where (\ref{M}) follows from the fact that size of each cache is $M$, (\ref{R}) follows from the fact that  for any demand in $\textbf{\textit{D}}$ the size of $X_{\textbf{\textit{d}}}$ is  at most $R(M)\leq R$,  (\ref{I.1})  follows from the fact that the file $W_{\textbf{\textit{d}}_{l}}$ can be computed from $X_{\textbf{\textit{d}}}$ and $Z_{l}$ by the user $U_{l}$, and (\ref{I.2}) follows from the fact that  $Z_{l}$ and $X_{\textbf{\textit{d}}}$ are functions of files $\{W_{1},\dots, W_{N}\}$. For a given cache size $M$, the smallest $R$ such that $(M, R)$ is achievable is called the exact rate memory tradeoff denoted by
\begin{equation}
R^{*}(M)=\min \{R: (M, R)\text{ is achievable}\}
\end{equation}

Maddah-Ali and Niesen in \cite{maddah2014fundamental} proposed a coding scheme with an uncoded placement phase and a coded delivery phase for the demands in $\textbf{\textit{D}}$ and demonstrated that the rate achieved by the proposed scheme is within a multiplicative gap of 12 from the optimal rate using cut set arguments. Several caching schemes were proposed in \cite{chen2014fundamental,amiri2016coded,yu2016exact,gomez2018fundamental,tian2016caching,vijith2019towards,kp2019fundamental,Shao_Gomez-Vukardebo_Zhang_Tian_2019,shao2020fundamental} to improve upon the rate achieved by the scheme proposed in \cite{maddah2014fundamental}. Despite several lower bounds on the achievable rates being proposed in \cite{ajaykrishnan2015critical,ghasemi2017improved,sengupta2015improved,  wang2018improved, yu2017characterizing}, the nature of the exact rate memory tradeoff is still elusive, except for the $(N, 2)$ cache network. 
The schemes proposed in \cite{maddah2014fundamental}, \cite{chen2014fundamental} provide a characterization of the exact rate memory tradeoff when $M\in\left[0,\frac{1}{K}\right]\cup  \left[\frac{N(K-1)}{K}, N\right]$.
For the special case of $N=K$, the schemes proposed in \cite{gomez2018fundamental}, \cite{kp2019fundamental} provide a characterization of the exact rate memory tradeoff when $M\in \left[\frac{1}{N}, \frac{1}{N(N-1)}\right]\cup \left[N-1-\frac{1}{N}, N-1 \right]$.
In a surprising result, Yu et al. \cite{yu2016exact} showed the existence of a universal scheme among caching schemes with an uncoded placement phase. These results are summarised in TABLE \ref{Table:Previous works}. In this paper, we consider the $(N, K)$ cache network where $N\leq K$  and $M\in \left[\frac{1}{K},\frac{N}{K}\right]$ and derive a set of new lower bounds for the demands in $\textbf{\textit{D}}$.
\begin{table*}[t]
	\centering
	\setlength{\tabcolsep}{2pt}
	\begin{tabular}{|c|c|c|c|}
		\hline
		 Caching Scheme& Cache Size $(M)$& Rate Memory Tradeoff&Condition\\
		\hline

		 Chen et al. \cite{chen2014fundamental}& $\big[0,\frac{1}{K}\big]$&$R^{*}(M)=N-NM$& $N\leq K$ \\
		\hline&&&\\[-1.2em]
		G{\'o}mez-Vilardeb{\'o} \cite{gomez2018fundamental}& $\big[\frac{1}{N},\frac{1}{N-1}\big]$&$R^{*}(M)=\dfrac{N^{2}-1}{N}-(N-1)M$&$N= K$\\[0.6em]
		\hline&&&\\[-1.2em]
		 Vijith et al. \cite{vijith2019towards}, \cite{kp2019fundamental} &$\left[N-1-\frac{1}{N},N-1\right]$&$R^{*}(M)=\frac{N+1}{N}-\frac{1}{N-1}M$ &$N=K$  \\[0.6em]
		\hline&&&\\[-1.2em]
		 $\begin{array}{c}
		\text{Maddah-Ali}\\ \text{\& Niesen \cite{maddah2014fundamental} }
		\end{array}$&$\left[\frac{N(K-1)}{K},N\right]$& $R^{*}(M)=1-\frac{1}{N}M$&-\\[0.6em]
		\hline&&&\\[-1.2em]
		 Yu et al. \cite{yu2016exact}&$[0,N]$& $\begin{array}{c}
		R(M)=R_{r}+(R_{r}-R_{r+1})\left(r-\frac{N}{K}M\right)\\
		\text{where }R(r)=\frac{{}^{K}\!C_{r+1}-{}^{K-N}\!C_{r+1}}{{}^{K}\!C_{r}}\\
		\text{and $r\in \{1,\dots, K\}$}
		\end{array}$ &$\begin{array}{c}
		\text{Optimal among}\\ \text{uncoded prefetching}\\ \text{schemes}
		\end{array}$\\[0.6em]
		
		\hline&&&\\[-1.2em]
		 This paper&$\left[\frac{1}{K},\frac{N}{K(N-1)}\right]$& $R^{*}(M)=\dfrac{KN-1}{K}-(N-1)M$ &$\big\lceil\frac{K+1}{2}\big\rceil\leq N \leq K$\\[0.6em]
		\hline
		
	\end{tabular}
\caption{Previous work in coded caching}
\label{Table:Previous works}
\end{table*}

The contributions of this paper are as follows:
\begin{itemize}
	\item For the case $\big\lceil\frac{K+1}{2}\big\rceil\leq N \leq K$, we derive a new lower bound and obtain a characterization of the exact rate memory tradeoff when $M\leq \frac{N}{K(N-1)}$. 
	\item For the case $1\leq N\leq \big\lceil\frac{K+1}{2}\big\rceil$, we derive a new lower bound which improves upon the previously known lower bounds.
\end{itemize}
Throughout this paper we use $[L]$ to represent the set $\{1,2,\dots,L\}$,  and $W_{[L]}$ to represent the set $\{W_{1}, W_{2},\dots, W_{L}\}$.

\section{Example networks}

In this section, we consider two examples to motivate the results we present in the paper. The $(3,4)$ network is an example for the case $\big\lceil\frac{K+1}{2}\big\rceil\leq N\leq K$ and the $(2,4)$ network is an example for the case $1\leq N\leq \big\lceil\frac{K+1}{2}\big\rceil$.

\subsection{Case I: The $(3,4)$ cache network}
\noindent Here, users $\{U_{1},U_{2},U_{3},U_{4}\}$ are connected to a server with three files $\{A,B,C\}$ (each of size $F$ bits). Each user $U_{k}$ has a cache $Z_{k}$ of size $MF$ bits. For a demand $\textbf{\textit{d}}$, we have:
\begin{lemma}
	\label{Theorem:(3,4)}
For the $(3,4)$ cache network, achievable memory rate pairs $(M,R)$ must satisfy the constraint
\begin{equation*}
8M+4R\geq 11
\end{equation*}
\end{lemma}
\begin{proof}

We have,
\begin{IEEEeqnarray*}{rll}
8M+4&R\overset{(a)}{\geq}  2H(Z_{1})+2H(Z_{2})+H(Z_{3})+3H(Z_{4})+2H(X_{(A,B,C,A)})+H(X_{(B,C,A,A)})+H(X_{(C,A,A,B)})\\
\overset{(b)}{\geq} & H(Z_{1},Z_{2},X_{(A,B,C,A)})+H(Z_{2},Z_{4},X_{(A,B,C,A)})+H(Z_{1},Z_{4},X_{(B,C,A,A)})+H(Z_{3},Z_{4},X_{(C,A,A,B)})\\
\overset{(c)}{=} & H(A,B,Z_{1},Z_{2},X_{(A,B,C,A)})+H(A,B,Z_{2},Z_{4},X_{(A,B,C,A)}) +H(A,B,Z_{1},Z_{4},X_{(B,C,A,A)})\\&+H(A,B,Z_{3},Z_{4},X_{(C,A,A,B)})\\
\overset{(b)}{\geq} &H(A,B,Z_{2},X_{(A,B,C,A)})+ H(A,B,Z_{1},Z_{2},Z_{4},X_{(A,B,C,A)}) +H(A,B,Z_{1},Z_{4},X_{(B,C,A,A)})\\&+H(A,B,Z_{3},Z_{4},X_{(C,A,A,B)})\\
\overset{}{\geq} & H(A,B,Z_{1},Z_{2},Z_{4}) +H(A,B,Z_{1},Z_{4},X_{(B,C,A,A)})+H(A,B,Z_{2},X_{(A,B,C,A)})\\&+H(A,B,Z_{3},Z_{4},X_{(C,A,A,B)})\\
\overset{(b)}{\geq} & H(A,B,Z_{1},Z_{2},Z_{4},X_{(B,C,A,A)})+H(A,B,Z_{1},Z_{4})+H(A,B,Z_{2},X_{(A,B,C,A)}) \\&+H(A,B,Z_{3},Z_{4},X_{(C,A,A,B)})\\
\overset{(c)}{=} & H(A,B,C,Z_{1},Z_{2},Z_{4},X_{(B,C,A,A)})+H(A,B,Z_{2},X_{(A,B,C,A)})+H(A,B,Z_{1},Z_{4}) \\&+H(A,B,Z_{3},Z_{4},X_{(C,A,A,B)})\\
\overset{(d)}{=} & H(A,B,C)+H(A,B,Z_{2},X_{(A,B,C,A)}) +H(A,B,Z_{1},Z_{4})+H(A,B,Z_{3},Z_{4},X_{(C,A,A,B)})\\
\overset{(b)}{\geq} & H(A,B,C)+H(A,B,Z_{2},X_{(A,B,C,A)}) +H(A,B,Z_{4})+H(A,B,Z_{1},Z_{3},Z_{4},X_{(C,A,A,B)})\\
\overset{(c)}{=} & H(A,B,C)+H(A,B,Z_{2},X_{(A,B,C,A)}) +H(A,B,Z_{4})+H(A,B,C,Z_{1},Z_{3},Z_{4},X_{(C,A,A,B)})\\
\overset{(d)}{=} & 2H(A,B,C)+H(A,B,Z_{2},X_{(A,B,C,A)})+H(A,B,Z_{4})\\
\overset{}{\geq} & 2H(A,B,C)+H(A,B,X_{(A,B,C,A)})+H(A,B,Z_{4})\\
\overset{(e)}{=} & 2H(A,B,C)+H(A,B,X_{(A,A,B,C)})+H(A,B,Z_{4})\\
\overset{(b)}{\geq} & 2H(A,B,C)+H(A,B)+H(A,B,Z_{4},X_{(A,A,B,C)})\\
\overset{(c)}{=} & 2H(A,B,C)+H(A,B)+H(A,B,C,Z_{4},X_{(A,A,B,C)})\\
\overset{(c)}{=} & 3H(A,B,C)+H(A,B)\geq 11,
\end{IEEEeqnarray*}
where 

\begin{tabular}{rl}
	$(a)$& follows from (\ref{R}) and (\ref{M}),\\ 
	$(b)$& follows from the submodularity property of entropy,\\
	$(c)$& follows from (\ref{I.1}), \\
	$(d)$& follows from (\ref{I.2}), \\
	$(e)$& follows from (\ref{symmetry}).
\end{tabular}

\end{proof}

 The above result improves upon the previous results from \cite{maddah2014fundamental,sengupta2015improved,gomez2018fundamental} and is summarised in TABLE \ref{Rate achieved for the $(3,4)$ cache network} and Fig. \ref{bound plot}.
\begin{table}[h]
	\centering
	\begin{tabular}{|c|c|c|c|}
		\hline
		Memory &Rate \cite{gomez2018fundamental}&Lower Bound\cite{maddah2014fundamental,sengupta2015improved}& New Lower Bound\\
		\hline&&&\\[-1.em]
		$\dfrac{1}{4}\leq M\leq \dfrac{3}{8}$&$\frac{11}{4}-2M$&$R\geq\max\left\{(3-3M),\left(\frac{8}{3}-2M\right)\right\}$ &$R\geq\frac{11}{4}-2M$\\[0.4em]
		\hline
	\end{tabular}
	\caption{Rate memory tradeoff for the $(3,4)$ cache network}
	\label{Rate achieved for the $(3,4)$ cache network}
\end{table}

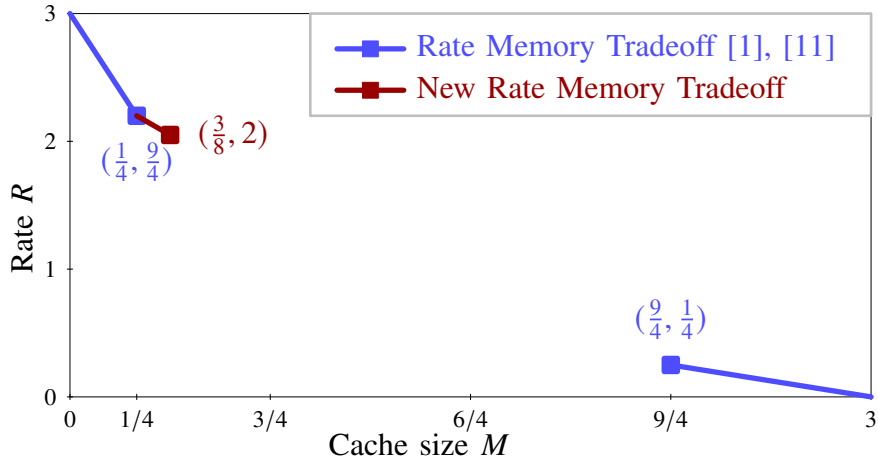
\begin{figure}[h]
	\centering
\input{34bound}
\caption{Rate memory tradeoff for the $(3,4)$ cache network}
\label{bound plot}
\end{figure}

\subsection{Case II: The $(2,4)$ cache network}
\noindent Here, users $\{U_{1},U_{2},U_{3},U_{4}\}$ are connected to a server with files $\{A,B\}$ (each of size $F$ bits). Each user $U_{k}$ has cache $Z_{k}$ of size  $MF$ bits. For a demand $\textbf{\textit{d}}$, we have:

\begin{lemma}
	For the $(2,4)$ cache network, achievable memory rate pairs $(M,R)$ must satisfy the constraint
	\begin{equation}
	8M+6R\geq 11.
	\end{equation}
	\end{lemma}
\begin{proof}

	We have,
	\begin{IEEEeqnarray*}{rl}
	8M&+6R\geq 2H(Z_{1})+H(Z_{2})+2H(Z_{3})+3H(Z_{4})+3H(X_{(A,B,A,A)})+2H(X_{(B,A,A,A)})+H(X_{(A,A,B,A)})\\
	\overset{(a)}{\geq}& H(Z_{1},X_{(A,B,A,A)})+H(Z_{3},X_{(A,B,A,A)})+H(Z_{4},X_{(A,B,A,A)})+H(Z_{3},X_{(B,A,A,A)})+H(Z_{4},X_{(B,A,A,A)})\\&+H(Z_{4},X_{(A,A,B,A)})+H(Z_{2})+H(Z_{1})\\
	\overset{(b)}{=}& H(A,Z_{1},X_{(A,B,A,A)})+H(A,Z_{3},X_{(A,B,A,A)})+H(A,Z_{4},X_{(A,B,A,A)})+H(A,Z_{3},X_{(B,A,A,A)})\\&+H(A,Z_{4},X_{(B,A,A,A)})+H(A,Z_{4},X_{(A,A,B,A)})+H(Z_{2})+H(Z_{1})\\
	\overset{(a)}{\geq}& H(A,Z_{1},Z_{3},Z_{4},X_{(A,B,A,A)})+2H(A,X_{(A,B,A,A)})+H(A,Z_{3},Z_{4},X_{(B,A,A,A)})+H(A,X_{(B,A,A,A)})\\&+H(A,Z_{4},X_{(A,A,B,A)})+H(Z_{2})+H(Z_{1})\\
	\overset{(c)}{=}& H(A,Z_{1},Z_{3},Z_{4},X_{(A,B,A,A)})+H(A,Z_{3},Z_{4},X_{(B,A,A,A)})+H(A,Z_{4},X_{(A,A,B,A)})+H(A,X_{(A,A,A,B)})\\&+H(A,X_{(A,B,A,A)})+H(A,X_{(B,A,A,A)})+H(Z_{2})+H(Z_{1})\\
	\overset{(a)}{\geq}& H(A,Z_{1},Z_{3},Z_{4},X_{(A,B,A,A)})+H(A,Z_{3},Z_{4},X_{(B,A,A,A)})+H(A,Z_{4},X_{(A,A,B,A)})+H(A,X_{(A,A,A,B)})\\&+H(A,Z_{2},X_{(A,B,A,A)})+H(A,Z_{1},X_{(B,A,A,A)})\\
	\overset{(b)}{=}& H(A,Z_{1},Z_{3},Z_{4},X_{(A,B,A,A)})+H(A,Z_{3},Z_{4},X_{(B,A,A,A)})+H(A,Z_{4},X_{(A,A,B,A)})+H(A,X_{(A,A,A,B)})\\&+H(A,B,Z_{2},X_{(A,B,A,A)})+H(A,B,Z_{1},X_{(B,A,A,A)})\\
	\overset{(d)}{=}& H(A,Z_{1},Z_{3},Z_{4},X_{(A,B,A,A)})+H(A,Z_{3},Z_{4},X_{(B,A,A,A)})+H(A,Z_{4},X_{(A,A,B,A)})+H(A,X_{(A,A,A,B)})\\&+2H(A,B)\\
	\overset{(a)}{\geq}& H(A,Z_{1},Z_{3},Z_{4},X_{(A,B,A,A)},X_{(B,A,A,A)})+H(A,Z_{3},Z_{4})+H(A,Z_{4},X_{(A,A,B,A)})+H(A,X_{(A,A,A,B)})\\&+2H(A,B)\\
	\overset{(b)}{=}& H(A,B,Z_{1},Z_{3},Z_{4},X_{(A,B,A,A)},X_{(B,A,A,A)})+H(A,Z_{3},Z_{4})+H(A,Z_{4},X_{(A,A,B,A)})+H(A,X_{(A,A,A,B)})\\&+2H(A,B)\\
	\overset{(d)}{=}& 3H(A,B)+H(A,Z_{3},Z_{4})+H(A,Z_{4},X_{(A,A,B,A)})+H(A,X_{(A,A,A,B)})\\
	\overset{(a)}{\geq}& 3H(A,B)+H(A,Z_{3},Z_{4},X_{(A,A,B,A)})+H(A,Z_{4})+H(A,X_{(A,A,A,B)})\\
	\overset{(b)}{=}& 3H(A, B)+H(A,B,Z_{3},Z_{4},X_{(A,A,B,A)})+H(A,Z_{4})+H(A,X_{(A,A,A,B)})\\
	\overset{(d)}{=}& 4H(A,B)+H(A,Z_{4})+H(A,X_{(A,A,A,B)})\\
	\overset{(a)}{\geq}& 4H(A,B)+H(A,Z_{4},X_{(A,A,A,B)})+H(A)\\
	\overset{(b)}{=}& 4H(A,B)+H(A,B,Z_{4},X_{(A,A,A,B)})+H(A)\\
	\overset{(d)}{=}& 5H(A,B)+H(A)\geq 11,
	\end{IEEEeqnarray*}
	where 
	
	\begin{tabular}{cl}
		$(a)$& follows from the submodularity property of entropy, \\
		$(b)$& follows from (\ref{I.1}),\\
		$(c)$& follows from (\ref{symmetry}), \\
		$(d)$& follows from (\ref{I.2}).
	\end{tabular}

\end{proof}
 The above result improves upon the previous results from \cite{maddah2014fundamental,sengupta2015improved,gomez2018fundamental} and is summarised in TABLE \ref{Rate achieved for the $(2,4)$ cache network} and Fig. \ref{(2,4)bound plot}.
\begin{table}[h]
	\centering
	\begin{tabular}{|c|c|c|c|}
		\hline
		Memory &Rate \cite{gomez2018fundamental}&Lower Bound\cite{maddah2014fundamental,sengupta2015improved}& New Lower Bound\\
		\hline&&&\\[-1.em]
		$\dfrac{1}{4}\leq M\leq \dfrac{1}{2}$&$\frac{32}{18}-\frac{10}{9}M $&$R\geq2-2M$ &$R\geq\frac{11}{6}-\frac{4}{3}M$\\[0.4em]
		\hline
	\end{tabular}
	\caption{Rate memory tradeoff for the $(2,4)$ cache network}
	\label{Rate achieved for the $(2,4)$ cache network}
\end{table}
\begin{figure}[h]
	\centering
	\input{24bound}
	\caption{Rate memory tradeoff for the $(2,4)$ cache network}
	\label{(2,4)bound plot}
\end{figure}
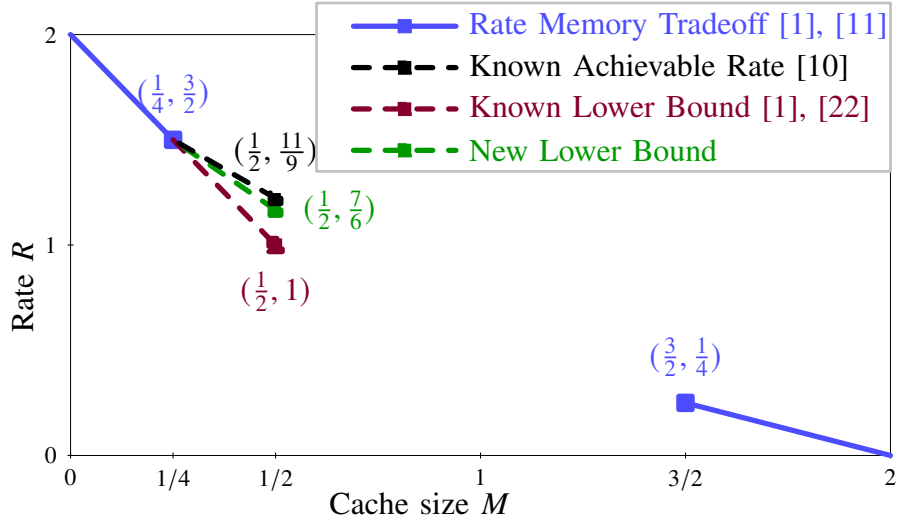
\begin{remark}
	It should be noted that, for the $(2,4)$ cache network, the bound $8M+6R\geq 11$ is already mentioned in \cite{tian2016symmetry}. We present the proof above, which can be extended to the $(N, K)$ cache network. 
\end{remark}

\section{New lower bounds}
In this section, we derive new lower bounds on the rate memory tradeoff  for the $(N,K)$ cache network where $N\leq K$ and cache size $M\in \left[\frac{1}{K},\frac{N}{K}\right]$. The key ideas we employ are identities (\ref{I.1}), (\ref{I.2}) and the properties of symmetric caching schemes stated in (\ref{symmetry}). As in Section II, we consider two cases, namely $\big\lceil\frac{K+1}{2}\big\rceil \leq N \leq K$ and $1\leq N \leq \big\lceil\frac{K+1}{2}\big\rceil$.

\subsection{Case I: $\big\lceil\frac{K+1}{2}\big\rceil \leq N \leq K$}
\noindent Consider the demand
\begin{equation}
\textbf{\textit{d}}_{1}=(W_{1},W_{2},\dots, W_{N}, W_{1}, W_{2},\dots, W_{K-N})
\end{equation}
Demands $\{\textbf{\textbf{\textit{d}}}_{l}:2\leq l\leq K\}$, are obtained from the demand $\textbf{\textit{d}}_{1}$ by cyclic left shifts as shown in TABLE \ref{demands}.
\begin{table}[!h]
	\centering
	\setlength{\tabcolsep}{2pt}
	\begin{tabular}{r|c|c|c|>{\columncolor{gray!25!white}}c|c|c|c|c|>{\columncolor{gray!25!white}}c|c|c|}
		\hhline{~-----------}   
		&Users &$\textbf{\textit{d}}_{1}$&\dots&$\textbf{\textit{d}}_{i}$&\dots&$\textbf{\textit{d}}_{N}$&$\textbf{\textit{d}}_{N+1}$&\dots&$\textbf{\textit{d}}_{N+i}$&\dots&$\textbf{\textit{d}}_{K}$\\ 
		\hhline{~-----------}   
		\ldelim\{{9}{5mm}[\parbox{3.5mm}{$ {\textbf{\textit{A}}_{i}}$}]&$ {U_{1}}$&$ {W_{1}}$& {\dots}& {$W_{i}$}& {\dots}& {$W_{N}$}& {$W_{1}$}& {\dots}& {$W_{i}$}& {\dots}& {$W_{K-N}$}\\  
		\hhline{~-----------}   
		& {$U_{2}$}& {$W_{2}$}& {\dots}& {$W_{i+1}$}& {\dots}& {$W_{1}$}& {$W_{2}$}& {\dots}& {$W_{i+1}$}& {\dots}& {$W_{1}$}\\  
		\hhline{~-----------}   
		& {\dots}& {\dots}& {\dots}& {\dots}& {\dots}& {\dots}& {\dots}& {\dots}& {\dots}& {\dots}& {\dots}\\

		\hhline{~-----------}    & {$U_{\substack{K-N-i+1}}$}& {$W_{\substack{K-N-i+1}}$}& {\dots}& {$W_{K-N}$}& {\dots}& {$W_{\substack{K-N-i}}$}& {$W_{\substack{K-N-i+1}}$}& {\dots}& {$W_{K-N}$}& {\dots}& {$W_{\substack{K-N-i}}$}\\  
		\hhline{~-----------}   
			\ldelim\{{5}{11mm}[\parbox{3.5mm}{ {${\textbf{\textit{C}}_{i}}$}}]& {$U_{\substack{K-N-i+2}}$}& {$W_{\substack{K-N-i+2}}$}& {\dots}& {$W_{K-N+1}$}& {\dots}& {$W_{\substack{K-N-i+1}}$}& {$W_{\substack{K-N-i+2}}$}& {\dots}& {$W_{1}$}& {\dots}& {$W_{\substack{K-N-i-1}}$}\\  
		\hhline{~-----------}   
		& {\dots}& {\dots}& {\dots}& {\dots}& {\dots}& {\dots}& {\dots}& {\dots}& {\dots}& {\dots}& {\dots}\\ 
		\hhline{~-----------}    & {$U_{\substack{K-N}}$}& {$W_{\substack{K-N}}$}& {\dots}& {$W_{\substack{K-N+i-1}}$}& {\dots}& {$W_{\substack{K-N-1}}$}& {$W_{\substack{K-N}}$}& {\dots}& {$W_{i-1}$}& {\dots}& {$W_{\substack{K-N-1}}$}\\  \hhline{~-----------}   
		& {\dots}& {\dots}& {\dots}& {\dots}& {\dots}& {\dots}& {\dots}& {\dots}& {\dots}& {\dots}& {\dots}\\ 
		\hhline{~-----------}    & {$U_{\substack{N-i}}$}& {$W_{\substack{N-i}}$}& {\dots}& {$W_{N-1}$}& {\dots}& {$W_{\substack{2N-K-i-1}}$}& {$W_{\substack{2N-K-i}}$}& {\dots}& {$W_{2N-K-1}$}& {\dots}& {$W_{\substack{N-i-1}}$}\\ 
		\hhline{~-----------} &$U_{\substack{N-i+1}}$&$W_{\substack{N-i+1}}$&\dots&\cellcolor{LimeGreen!50!white}{$W_{N}$}&\dots&$W_{\substack{2N-K-i}}$&$W_{\substack{2N-K-i+1}}$&\dots&$W_{2N-K}$&\dots&$W_{\substack{N-i}}$\\  \hhline{~-----------}   
		&\dots&\dots&\dots&\dots&\dots&\dots&\dots&\dots&\dots&\dots&\dots\\ 
		
		\hhline{~-----------}    &$U_{\substack{N}}$&$W_{\substack{N}}$&\dots&$W_{i-1}$&\dots&$W_{\substack{2N-K}}$&$W_{\substack{2N-K+1}}$&\dots&$W_{\substack{2N-K+i}}$&\dots&$W_{\substack{N-1}}$\\  \hhline{~-----------}   
		\ldelim\{{6}{5mm}[\parbox{3.5mm}{{ {${\textbf{\textit{E}}}$}}}]& {$U_{\substack{N+1}}$}& {$W_{\substack{1}}$}& {\dots}& {$W_{i}$}& {\dots}& {$W_{\substack{2N-K+1}}$}& {$W_{\substack{2N-K+2}}$}& {\dots}& {$W_{\substack{2N-K+i+1}}$}& {\dots}& {$W_{\substack{N}}$}\\  \hhline{~-----------}   
		& {\dots}& {\dots}& {\dots}& {\dots}& {\dots}& {\dots}& {\dots}& {\dots}& {\dots}& {\dots}& {\dots}\\ 
	
		\hhline{~-----------} & {$U_{\substack{K-i+1}}$}& {$W_{\substack{K-N-i+1}}$}& {\dots}& {$W_{K-N}$}& {\dots}& {$W_{\substack{N-i}}$}& {$W_{\substack{N-i+1}}$}& {\dots}&\cellcolor{LimeGreen!50!white} {$W_{\substack{N}}$}& {\dots}& {$W_{\substack{K-N-i}}$}\\  
		\hhline{~-----------}   
		\!\!\ldelim\{{3}{11mm}[\parbox{3.5mm}{ {${\textbf{\textit{B}}_{i}}$}}]& {$U_{\substack{K-i+2}}$}& {$W_{\substack{K-N-i+2}}$}& {\dots}& {$W_{1}$}& {\dots}& {$W_{\substack{N-i+1}}$}& {$W_{\substack{N-i+2}}$}& {\dots}& {$W_{\substack{1}}$}& {\dots}& {$W_{\substack{K-N-i+1}}$}\\  
		\hhline{~-----------}   
		& {\dots}& {\dots}& {\dots}& {\dots}& {\dots}& {\dots}& {\dots}& {\dots}& {\dots}& {\dots}& {\dots}\\ 
		
		\hhline{~-----------}    & {$U_{K}$}& {$W_{\substack{K-N}}$}& {\dots}& {$W_{i-1}$}& {\dots}& {$W_{\substack{N-1}}$}& {$W_{\substack{N}}$}& {\dots}& {$W_{\substack{i-1}}$}& {\dots}& {$W_{\substack{K-N-1}}$}\\  
		\hhline{~-----------}
	\end{tabular}
	\caption{The set of demands $\{\textbf{\textit{d}}_{l}: 1\leq l\leq K\}$}
	\label{demands}
\end{table}
\noindent For the demand  $\textbf{\textit{d}}_{l}$, let $X_{\textbf{\textit{d}}_{l}}$ denote the set of packets broadcast by the server. Consider the user index $\overline{l}$ defined as
\begin{equation}
\label{case 1 ZWN}
\overline{l}=\left\{\begin{aligned}
{N+1-l}, &\text{ for $1\leq l\leq N$}\\
{K+N+1-l}, &\text{ for $N+1\leq l\leq K$}
\end{aligned}\right.
\end{equation}
It can be noted that in demand $\textbf{\textit{d}}_{l}$, the user $U_{\overline{l}}$ requires the file $W_{N}$. For $\textbf{\textit{S}}\subseteq \{{U}_{1},\dots,{U}_{K}\}$, let $Z_{\textbf{\textit{S}}}$ denote the cache contents of all the users in set ${\textbf{\textit{S}}}$.

The following lemma are easy to obtain:
\begin{lemma}
	\label{Lemma repeated}
	 For $\textbf{\textit{S}},\textbf{\textit{T}}\subset\{U_{1},\dots, U_{K}\}\setminus\{U_{\overline{l}}\}$, we have the identity
	\begin{align*}
	H(W_{[N-1]},Z_{\textbf{\textit{S}}},Z_{\overline{l}})+H(&W_{[N-1]},Z_{\textbf{\textit{T}}},X_{\textbf{\textit{d}}_{l}})\geq H(W_{[N-1]},Z_{\textbf{\textit{S}}\cap \textbf{\textit{T}}})+N,
	\end{align*}
\end{lemma}
\begin{proof}
	We have,
	\begin{align*}
	H(W_{[N-1]},Z_{\textbf{\textit{S}}},Z_{\overline{l}})+H(W_{[N-1]},Z_{\textbf{\textit{T}}},X_{\textbf{\textit{d}}_{l}})\overset{(a)}{\geq} &H(W_{[N-1]},Z_{\textbf{\textit{S}}\cap \textbf{\textit{T}}})+H(W_{[N-1]},Z_{\textbf{\textit{S}}\cup \textbf{\textit{T}}},Z_{\overline{l}},X_{\textbf{\textit{d}}_{l}})\\
	\overset{(b)}{=}&H(W_{[N-1]},Z_{\textbf{\textit{S}}\cap \textbf{\textit{T}}})+H(W_{[N-1]},W_{N},Z_{\textbf{\textit{S}}\cup \textbf{\textit{T}}},Z_{\overline{l}},X_{\textbf{\textit{d}}_{l}})\\
	\overset{(c)}{=}&H(W_{[N-1]},Z_{\textbf{\textit{S}}\cap \textbf{\textit{T}}})+H(W_{[N]})\\=&H(W_{[N-1]},Z_{\textbf{\textit{S}}\cap \textbf{\textit{T}}})+N
	\end{align*}
	where 
	
	\begin{tabular}{cl}
		$(a)$& follows from the submodularity property of entropy,\\
		$(b)$& follows from (\ref{I.1}),\\
		$(c)$& follows from (\ref{I.2}).
	\end{tabular}

\end{proof}
\begin{lemma}
	\label{Lemma sum increment}
	For a sequence of sets $\textbf{\textit{S}}_{i}\subset\{U_{1},\dots,U_{K}\}\setminus\{U_{\overline{i}}\}$, such that $\textbf{\textit{S}}_{i}=\textbf{\textit{S}}_{i+1}\cup\{U_{\overline{i+1}}\}$, we have
	\begin{align*}
	H(W_{[N-1]},Z_{\textbf{\textit{S}}_{l}})+\sum_{i=l+1}^{j}H(W_{[N-1]},Z_{\textbf{\textit{S}}_{i}},X_{\textbf{\textit{d}}_{i}})\geq& (j-l)N+H(W_{[N-1]},Z_{\textbf{\textit{S}}_{j}})
	\end{align*}
\end{lemma}
\begin{proof}
	We have,
	\begin{align*}
	H(W_{[N-1]},Z_{\textbf{\textit{S}}_{l}})+&\sum_{i=l+1}^{j}H(W_{[N-1]},Z_{\textbf{\textit{S}}_{i}},X_{\textbf{\textit{d}}_{i}})\\=&H(W_{[N-1]},Z_{\textbf{\textit{S}}_{l}}) +H(W_{[N-1]},Z_{\textbf{\textit{S}}_{l+1}},X_{\textbf{\textit{d}}_{l+1}})+\sum_{i=l+2}^{j}H(W_{[N-1]},Z_{\textbf{\textit{S}}_{i}},X_{\textbf{\textit{d}}_{i}})\\
	\overset{(a)}{=}&\Big[H(W_{[N-1]},Z_{\textbf{\textit{S}}_{l+1}},Z_{\overline{l+1}})+H(W_{[N-1]},Z_{\textbf{\textit{S}}_{l+1}},X_{\textbf{\textit{d}}_{l+1}})\Big] +\sum_{i=l+2}^{j}H(W_{[N-1]},Z_{\textbf{\textit{S}}_{i}},X_{\textbf{\textit{d}}_{i}})\\
	\overset{(b)}{\geq}&N+\Big[H(W_{[N-1]},Z_{\textbf{\textit{S}}_{l+1}})+H(W_{[N-1]},Z_{\textbf{\textit{S}}_{l+2}},X_{\textbf{\textit{d}}_{l+2}})\Big] +\sum_{i=l+3}^{j}H(W_{[N-1]},Z_{\textbf{\textit{S}}_{i}},X_{\textbf{\textit{d}}_{i}})\\
	\overset{(c)}{\geq}&2N+H(W_{[N-1]},Z_{\textbf{\textit{S}}_{l+2}})+H(W_{[N-1]},Z_{\textbf{\textit{S}}_{l+3}},X_{\textbf{\textit{d}}_{l+3}}) +\sum_{i=l+4}^{j}H(W_{[N-1]},Z_{\textbf{\textit{S}}_{i}},X_{\textbf{\textit{d}}_{i}})\\
	\overset{(d)}{\geq}&(j-l)N+H(W_{[N-1]},Z_{\textbf{\textit{S}}_{j}})
	\end{align*}
	where 
	
	\begin{tabular}{cl}
		$(a)$& follows from definition of set $\textbf{\textit{S}}_{i}$,\\
		$(b)$& follows from Lemma \ref{Lemma repeated} with $\textbf{\textit{S}}\cup \{U_{\overline{l}}\}=\textbf{\textit{S}}_{l+1}$ and $\textbf{\textit{T}}=\textbf{\textit{S}}_{l+2}$,\\ 
		$(c)$& follows from Lemma \ref{Lemma repeated} with $\textbf{\textit{S}}\cup \{U_{\overline{l}}\}=\textbf{\textit{S}}_{l+2}$ and $\textbf{\textit{T}}=\textbf{\textit{S}}_{l+3}$,\\ 
		$(d)$& follows from repeated use of Lemma \ref{Lemma repeated} with $\textbf{\textit{S}}\cup \{U_{\overline{l}}\}=\textbf{\textit{S}}_{i}$ and $\textbf{\textit{T}}=\textbf{\textit{S}}_{i+1}$\\& for $l+3\leq i\leq j$.
	\end{tabular}

\end{proof}
\noindent In a similar fashion, for a sequence of sets $\textbf{\textit{T}}_{i}\subset\{U_{1},\dots,U_{K}\}\setminus \{U_{\overline{i}}\}$, such that $\textbf{\textit{T}}_{i+1}=\textbf{\textit{T}}_{i}\cup \{U_{\overline{i}}\}$, we can obtain
\begin{equation}
\label{Lemma sum decrement}
H(W_{[N-1]},Z_{\textbf{\textit{T}}_{j}},Z_{\overline{j}})+\sum_{i=l}^{j}H(W_{[N-1]},Z_{\textbf{\textit{T}}_{i}},X_{\textbf{\textit{d}}_{i}})\geq (j-l+1)N+H(W_{[N-1]},Z_{\textbf{\textit{T}}_{l}})
\end{equation}

\noindent For $1\leq i\leq N$, consider the sets of users as shown below:
\begin{table}[h!]
	\centering
	 \setlength{\tabcolsep}{1.5pt}
	\begin{tabular}{|c|c|c|c|}
		\hline
		Set &Users &Number& {Files Requested in Demand $\textbf{\textit{d}}_{i}$}\\
		\hline
		 {${\textbf{\textit{A}}_{i}}$}& {$U_{1},\dots, U_{N-i}$}& $N-i$& {$W_{i},\dots, W_{N-1}$}\\
		\hline
		 {${\textbf{\textit{B}}_{i}}$}& {$U_{K+2-i},\dots, U_{K}$}&$i-1$& {$W_{1},\dots, W_{i-1}$}\\
		\hline
		 {${\textbf{\textit{C}}_{i}}$}& {$U_{K+2-N-i},\dots, U_{N-i}$}&$\begin{array}{c}
		2N-K-1
		\end{array}$& {$W_{K-N+1},\dots,W_{N-1}$}\\
		\hline
		 {${\textbf{\textit{E}}}$}& {$U_{N+1},\dots, U_{K}$}&$K-N$& {$W_{1},\dots, W_{K-N}$}\\
		\hline
	\end{tabular}

\end{table}

\noindent These sets are also indicated in TABLE \ref{demands}. Note that
\begin{align}
\label{case 1 Null}
\textbf{\textit{A}}_{N}=\textbf{\textit{B}}_{1}=\textbf{\textit{C}}_{N}&=\phi\\
\label{case 1 Set A}
\textbf{\textit{A}}_{i+1}\cup \{U_{\overline{i+1}}\}&=\textbf{\textit{A}}_{i}\\
\label{case 1 Set B}
\textbf{\textit{B}}_{i}\cup \{U_{\overline{N+i}}\}&=\textbf{\textit{B}}_{i+1}\\
\label{case 1 B K-N}
\textbf{\textit{B}}_{K-N}\cup \{U_{\overline{K}}\}=\textbf{\textit{B}}_{K-N+1}&= \textbf{\textit{E}}\\
\label{case 1 AD}
\textbf{\textit{A}}_{i}\cap \textbf{\textit{C}}_{i}&=\textbf{\textit{C}}_{i}\\
\label{case 1 BD}
\textbf{\textit{B}}_{i}\cap  \textbf{\textit{E}}&=\left\{ \text{ }\begin{aligned}[cl]
\textbf{\textit{B}}_{i} &\text{         when $1\leq i\leq K-N$}\\
\textbf{\textit{E}}&\text{              when $K-N+1\leq i\leq N$}
\end{aligned}
\right.
\end{align}

It can  be noted that in the demands $\textbf{\textit{d}}_{i}$ and $ \textbf{\textit{d}}_{N+i}$, the users in set $\textbf{\textit{B}}_{i}$ are requesting for the same set of files $\{W_{1},\dots, W_{i-1}\}$ (for $1\leq i\leq N$). Thus, from (\ref{symmetry}) we have
\begin{equation}
\label{case 1 symmetry}
H(W_{[i-1]},Z_{\textbf{\textit{B}}_{i}},X_{\textbf{\textit{d}}_{i}})=H(W_{[i-1]},Z_{\textbf{\textit{B}}_{i}},X_{\textbf{\textit{d}}_{N+i}})
\end{equation}
Note that $\mid \textbf{\textit{A}}_{i}\cup \textbf{\textit{B}}_{i}\mid=\mid \textbf{\textit{C}}_{i}\cup\textbf{\textit{E}}\mid=N-1$. Thus, we have
\begin{align}
\label{case1:A,B}
(N-1)M+R\geq H(Z_{\textbf{\textit{A}}_{i}\cup \textbf{\textit{B}}_{i}})+H(X_{\textbf{\textit{d}}_{i}})\geq H(Z_{\textbf{\textit{A}}_{i}\cup \textbf{\textit{B}}_{i}},X_{\textbf{\textit{d}}_{i}})
\end{align}
Similarly,
\begin{align}
\label{case1:C,J}
(N-1)M+R\geq H(Z_{\textbf{\textit{C}}_{i}\cup \textbf{\textit{E}}})+H(X_{\textbf{\textit{d}}_{i}})\geq H(Z_{\textbf{\textit{C}}_{i}\cup \textbf{\textit{E}}},X_{\textbf{\textit{d}}_{i}})
\end{align}
\noindent Now, we have the following result:
\begin{theorem}
	\label{Main_Result_1}
	For the $(N,K)$ cache network, when $\big\lceil\frac{K+1}{2}\big\rceil\leq N\leq K$, achievable memory rate pairs $(M,R)$ must satisfy the constraint
	\begin{equation*}
	K(N-1)M+KR\geq KN-1.
	\end{equation*}
\end{theorem}
\begin{proof}
	We have,	
	\begin{IEEEeqnarray*}{rl}
	K(N-1)&M+KR=\sum_{i=1}^{K-N}\Big[(N-1)M+R+(N-1)M+R\Big]+\sum_{i=K-N+1}^{N}\Big[(N-1)M+R\Big]\\
	\overset{(a)}{\geq}&\sum_{i=1}^{K-N}\Big[H(Z_{\textbf{\textit{A}}_{i}\cup\textbf{\textit{B}}_{i}},X_{\textbf{\textit{d}}_i})+H(Z_{\textbf{\textit{C}}_{i}\cup \textbf{\textit{E}}},X_{\textbf{\textit{d}}_i})\Big]+\sum_{i=K-N+1}^{N}H(Z_{\textbf{\textit{A}}_{i}\cup\textbf{\textit{B}}_{i}},X_{\textbf{\textit{d}}_i})\\
	\overset{(b)}{=}&\sum_{i=1}^{K-N}\Big[H(W_{[N-1]},Z_{\textbf{\textit{A}}_{i}\cup\textbf{\textit{B}}_{i}},X_{\textbf{\textit{d}}_i})+H(W_{[N-1]},Z_{\textbf{\textit{C}}_{i}\cup \textbf{\textit{E}}},X_{\textbf{\textit{d}}_i})\Big]+\sum_{i=K-N+1}^{N}H(W_{[N-1]},Z_{\textbf{\textit{A}}_{i}\cup\textbf{\textit{B}}_{i}},X_{\textbf{\textit{d}}_i})\\
	\overset{(c)}{\geq}& \sum_{i=1}^{K-N}\Big[ H(W_{[N-1]},Z_{\textbf{\textit{A}}_{i}\cup \textbf{\textit{E}}},X_{\textbf{\textit{d}}_i})+H(W_{[N-1]},Z_{\textbf{\textit{B}}_{i}\cup\textbf{\textit{C}}_{i}},X_{\textbf{\textit{d}}_i})\Big]+\sum_{i=K-N+1}^{N}H(W_{[N-1]},Z_{\textbf{\textit{A}}_{i}\cup\textbf{\textit{B}}_{i}},X_{\textbf{\textit{d}}_i})\\
	\overset{}{\geq}&\left[H(W_{[N-1]},Z_{\textbf{\textit{A}}_{1}\cup \textbf{\textit{E}}})+ \sum_{i=2}^{K-N}H(W_{[N-1]},Z_{\textbf{\textit{A}}_{i}\cup \textbf{\textit{E}}},X_{\textbf{\textit{d}}_i})\right]+\sum_{i=K-N+1}^{N}H(W_{[N-1]},Z_{\textbf{\textit{A}}_{i}\cup\textbf{\textit{B}}_{i}},X_{\textbf{\textit{d}}_i})\\&+\sum_{i=1}^{K-N}H(W_{[N-1]},Z_{\textbf{\textit{B}}_{i}\cup\textbf{\textit{C}}_{i}},X_{\textbf{\textit{d}}_i})\\
	\overset{(d)}{\geq}&(K-N-1)N+\left[H(W_{[N-1]},Z_{\textbf{\textit{A}}_{K-N}\cup \textbf{\textit{E}}})+\sum_{i=K-N+1}^{N}H(W_{[N-1]},Z_{\textbf{\textit{A}}_{i}\cup\textbf{\textit{B}}_{i}},X_{\textbf{\textit{d}}_i})\right]\\&+\sum_{i=1}^{K-N}H(W_{[N-1]},Z_{\textbf{\textit{B}}_{i}\cup\textbf{\textit{C}}_{i}},X_{\textbf{\textit{d}}_i})\\
	\overset{(e)}{=}&(K-N-1)N+\left[H(W_{[N-1]},Z_{\textbf{\textit{A}}_{K-N}\cup \textbf{\textit{B}}_{K-N+1}})+\sum_{i=K-N+1}^{N}H(W_{[N-1]},Z_{\textbf{\textit{A}}_{i}\cup\textbf{\textit{B}}_{i}},X_{\textbf{\textit{d}}_i})\right]\\&+\sum_{i=1}^{K-N}H(W_{[N-1]},Z_{\textbf{\textit{B}}_{i}\cup\textbf{\textit{C}}_{i}},X_{\textbf{\textit{d}}_i})\\
	\overset{(f)}{\geq}&(N-1)N+H(W_{[N-1]},Z_{\textbf{\textit{A}}_{N}\cup \textbf{\textit{B}}_{K-N+1}}) +\sum_{i=1}^{K-N}H(W_{[N-1]},Z_{\textbf{\textit{B}}_{i}\cup\textbf{\textit{C}}_{i}},X_{\textbf{\textit{d}}_i})\\
	\overset{(g)}{=}&(N-1)N+H(W_{[N-1]},Z_{\textbf{\textit{B}}_{K-N+1}}) +\sum_{i=1}^{K-N}H(W_{[N-1]},Z_{\textbf{\textit{B}}_{i}\cup\textbf{\textit{C}}_{i}},X_{\textbf{\textit{d}}_i})\\
	\overset{}{\geq}&(N-1)N+H(W_{[N-1]},Z_{\textbf{\textit{B}}_{K-N+1}}) +\sum_{i=1}^{K-N}H(W_{[N-1]},Z_{\textbf{\textit{B}}_{i}},X_{\textbf{\textit{d}}_i})\\
	\overset{(h)}{=}&(N-1)N+\left[H(W_{[N-1]},Z_{\textbf{\textit{B}}_{K-N+1}}) +\sum_{i=1}^{K-N}H(W_{[N-1]},Z_{\textbf{\textit{B}}_{i}},X_{\textbf{\textit{d}}_{N+i}})\right]\\
	\overset{(i)}{\geq}&(K-1)N+H(W_{[N-1]},Z_{\textbf{\textit{B}}_{1}})\\
	\overset{(g)}{=}&(K-1)N+H(W_{[N-1]})
	\geq KN-1
	\end{IEEEeqnarray*}
		where 

	\begin{tabular}{cl}
		$(a)$& follows from (\ref{case1:A,B}) and (\ref{case1:C,J}),\\
		$(b)$& follows from (\ref{I.1}) and definition of sets $\textbf{\textit{A}}_{i}$, $\textbf{\textit{B}}_{i}$, $\textbf{\textit{C}}_{i}$ and $ \textbf{\textit{E}}$,\\
		$(c)$& follows from the submodularity property of entropy, and the facts that\\&  $\textbf{\textit{A}}_{i}\cap\textbf{\textit{C}}_{i}=\textbf{\textit{C}}_{i}$ and $\textbf{\textit{B}}_{i}\cap\textbf{\textit{E}}=\textbf{\textit{B}}_{i}$ for $1\leq i\leq K-N$ (refer (\ref{case 1 AD}) and (\ref{case 1 BD})),\\
		$(d)$& follows from Lemma \ref{Lemma sum increment} with $\textbf{\textit{S}}_{i}=\textbf{\textit{A}}_{i}\cup \textbf{\textit{E}}$, $l=1$, $j=K-N$ and (\ref{case 1 Set A}),\\
		$(e)$& follows from (\ref{case 1 B K-N})\\
		$(f)$& follows from Lemma \ref{Lemma sum increment} with $\textbf{\textit{S}}_{i}=\textbf{\textit{A}}_{i}\cup \textbf{\textit{B}}_{i}$, $l=K-N$, $j=N$ and (\ref{case 1 Set A}),\\
		$(g)$& follows from (\ref{case 1 Null}),\\
		$(h)$& follows from (\ref{case 1 symmetry}),\\
		$(i)$& follows from (\ref{Lemma sum decrement}) with $\textbf{\textit{T}}_{i}=\textbf{\textit{B}}_{i}$, $l=1$, $j=K-N$ and (\ref{case 1 Set B}).
	\end{tabular}

\end{proof}
\subsection{Case II: $1\leq N\leq \big\lceil\frac{K+1}{2}\big\rceil$ }
\noindent Consider the demand 
\begin{equation}
\textbf{\textit{d}}_{1}=(W_{1},W_{2},\dots, W_{N},W_{1},W_{2},\dots, W_{N-1}, W_{1},W_{1},\dots,W_{1})
\end{equation}
Demands $\{\textbf{\textbf{\textit{d}}}_{l}:2\leq l\leq K\}$, are obtained from the demand $\textbf{\textit{d}}_{1}$ by cyclic left shifts as shown in TABLE \ref{demands case 2}. 
\begin{table}[!h]
	\centering
	\setlength{\tabcolsep}{1.5pt}
	\begin{tabular}{r|c|c|c|>{\columncolor{gray!30!white}}c|c|c|c|c|>{\columncolor{gray!30!white}}c|c|c|c|c|>{\columncolor{gray!30!white}}c|c|c|}
		\hhline{~----------------}
		&Users&$\textbf{\textit{d}}_{1}$&\dots&$\textbf{\textit{d}}_{i}$&\dots&$\textbf{\textit{d}}_{N}$&$\textbf{\textit{d}}_{N+1}$&\dots&$\textbf{\textit{d}}_{N+i}$&\dots&$\textbf{\textit{d}}_{2N-1}$&$\textbf{\textit{d}}_{2N}$&\dots&$\textbf{\textit{d}}_{j}$&\dots&$\textbf{\textit{d}}_{K}$\\
		\hhline{~----------------}
		\ldelim\{{4.3}{5.5mm}[\parbox{3.5mm}{${\textbf{\textit{J}}_{i}}$}]\ldelim\{{3}{5.5mm}[\parbox{3.5mm}{${\textbf{\textit{A}}_{i}}$}]&$U_{1}$&$W_{1}$&\dots&$W_{i}$&\dots&$W_{N}$&$W_{1}$&\dots&$W_{i}$&\dots&$W_{N-1}$&$W_{1}$&\dots&$W_{1}$&\dots&$W_{1}$\\
		\hhline{~----------------}
		&\dots&\dots&\dots&\dots&\dots&\dots&\dots&\dots&\dots&\dots&\dots&\dots&\dots&\dots&\dots&\dots\\
		\hhline{~----------------}
		&$U_{N-i}$&$W_{N-i}$&\dots&$W_{N-1}$&\dots&$W_{\substack{N-i-1}}$&$W_{N-i}$&\dots&$W_{N-1}$&\dots&$W_{1}$&$W_{1}$&\dots&$W_{1}$&\dots&$W_{i-1}$\\
		\hhline{~----------------}
		&$U_{\substack{N-i+1}}$&$W_{\substack{N-i+1}}$&\dots&\cellcolor{LimeGreen!50!white}$W_{N}$&\dots&$W_{N-i}$&$W_{\substack{N-i+1}}$&\dots&$W_{1}$&\dots&$W_{1}$&$W_{1}$&\dots&$W_{1}$&\dots&$W_{N-i}$\\
		\hhline{~----------------}

		&\dots&\dots&\dots&\dots&\dots&\dots&\dots&\dots&\dots&\dots&\dots&\dots&\dots&\dots&\dots&\dots\\
		\hhline{~----------------}
		\ldelim\{{3}{5.5mm}[\parbox{3.5mm}{${\textbf{\textit{F}}_{i}}$}]&$U_{N+1}$&$W_{1}$&\dots&$W_{i}$&\dots&$W_{1}$&$W_{1}$&\dots&$W_{1}$&\dots&$W_{1}$&$W_{1}$&\dots&$W_{1}$&\dots&$W_{N}$\\
		\hhline{~----------------}
		&\dots&\dots&\dots&\dots&\dots&\dots&\dots&\dots&\dots&\dots&\dots&\dots&\dots&\dots&\dots&\dots\\
		\hhline{~----------------}
		&$U_{2N-i}$&$W_{N-i}$&\dots&$W_{N-1}$&\dots&$W_{1}$&$W_{1}$&\dots&$W_{1}$&\dots&$W_{1}$&$W_{1}$&\dots&$W_{1}$&\dots&$W_{\substack{N-i-1}}$\\
		\hhline{~----------------}
		\ldelim\{{12}{5.5mm}[\parbox{3.5mm}{${\textbf{\textit{G}}_{i}}$}]&$U_{\substack{2N-i+1}}$&$W_{\substack{N-i+1}}$&\dots&$W_{1}$&\dots&$W_{1}$&$W_{1}$&\dots&$W_{1}$&\dots&$W_{1}$&$W_{1}$&\dots&$W_{1}$&\dots&$W_{N-i}$\\

					\hhline{~----------------}
					&\dots&\dots&\dots&\dots&\dots&\dots&\dots&\dots&\dots&\dots&\dots&\dots&\dots&\dots&\dots&\dots\\
		\hhline{~----------------}
		&$U_{\substack{K-j+2}}$&$W_{1}$&\dots&$W_{1}$&\dots&$W_{1}$&$W_{1}$&\dots&$W_{1}$&\dots&$W_{1}$&$W_{1}$&\dots&$W_{1}$&\dots&$W_{1}$\\
		\hhline{~----------------}
		\ldelim\{{3.3}{16.5mm}[\parbox{3.5mm}{${\textbf{\textit{S}}_{j}}$}]&$U_{\substack{K-j+3}}$&$W_{1}$&\dots&$W_{1}$&\dots&$W_{1}$&$W_{1}$&\dots&$W_{1}$&\dots&$W_{1}$&$W_{1}$&\dots&$W_{2}$&\dots&$W_{1}$\\
		\hhline{~----------------}
		&\dots&\dots&\dots&\dots&\dots&\dots&\dots&\dots&\dots&\dots&\dots&\dots&\dots&\dots&\dots&\dots\\
		\hhline{~----------------}
		
		&$U_{\substack{K+N-j+1}}$&$W_{1}$&\dots&$W_{1}$&\dots&$W_{1}$&$W_{1}$&\dots&$W_{1}$&\dots&$W_{1}$&$W_{1}$&\dots&\cellcolor{LimeGreen!50!white}$W_{N}$&\dots&$W_{1}$\\
		\hhline{~----------------}

		\ldelim\{{3.3}{16.5mm}[\parbox{3.5mm}{${\textbf{\textit{P}}_{j}}$}]&$U_{\substack{K+N-j+2}}$&$W_{1}$&\dots&$W_{1}$&\dots&$W_{1}$&$W_{1}$&\dots&$W_{1}$&\dots&$W_{1}$&$W_{1}$&\dots&$W_{1}$&\dots&$W_{1}$\\
		\hhline{~----------------}
		&\dots&\dots&\dots&\dots&\dots&\dots&\dots&\dots&\dots&\dots&\dots&\dots&\dots&\dots&\dots&\dots\\
		\hhline{~----------------}
		&$U_{\substack{K+2N-j}}$&$W_{1}$&\dots&$W_{1}$&\dots&$W_{1}$&$W_{1}$&\dots&$W_{1}$&\dots&$W_{1}$&$W_{1}$&\dots&$W_{N-1}$&\dots&$W_{1}$\\
		\hhline{~----------------}
		
		\ldelim\{{7}{16.5mm}[\parbox{3.5mm}{${\textbf{\textit{Q}}_{j}}$}]&$U_{\substack{K+2N-j+1}}$&$W_{1}$&\dots&$W_{1}$&\dots&$W_{1}$&$W_{1}$&\dots&$W_{1}$&\dots&$W_{1}$&$W_{1}$&\dots&$W_{1}$&\dots&$W_{1}$\\
		\hhline{~----------------}
		&\dots&\dots&\dots&\dots&\dots&\dots&\dots&\dots&\dots&\dots&\dots&\dots&\dots&\dots&\dots&\dots\\
		\hhline{~----------------}

		&$U_{\substack{K-i+1}}$&$W_{1}$&\dots&$W_{1}$&\dots&$W_{N-i}$&$W_{\substack{N-i+1}}$&\dots&\cellcolor{LimeGreen!50!white}$W_{N}$&\dots&$W_{N-i}$&$W_{\substack{N-i+1}}$&\dots&$W_{1}$&\dots&$W_{1}$\\
		\hhline{~----------------}
		\ldelim\{{4}{11mm}[\parbox{3.5mm}{${\textbf{\textit{B}}_{i}}$}]&$U_{\substack{K-i+2}}$&$W_{1}$&\dots&$W_{1}$&\dots&$W_{\substack{N-i+1}}$&$W_{\substack{N-i+2}}$&\dots&$W_{1}$&\dots&$W_{\substack{N-i+1}}$&$W_{\substack{N-i+2}}$&\dots&$W_{1}$&\dots&$W_{1}$\\
		\hhline{~----------------}
		\ldelim\{{3}{5.5mm}[\parbox{3.5mm}{${\textbf{\textit{K}}_{i}}$}]&$U_{\substack{K-i+3}}$&$W_{1}$&\dots&$W_{2}$&\dots&$W_{\substack{N-i+2}}$&$W_{\substack{N-i+3}}$&\dots&$W_{2}$&\dots&$W_{\substack{N-i+2}}$&$W_{\substack{N-i+3}}$&\dots&$W_{1}$&\dots&$W_{1}$\\
		\hhline{~----------------}
		&\dots&\dots&\dots&\dots&\dots&\dots&\dots&\dots&\dots&\dots&\dots&\dots&\dots&\dots&\dots&\dots\\
		\hhline{~----------------}
		&$U_{K}$&$W_{1}$&\dots&$W_{i-1}$&\dots&$W_{N-1}$&$W_{N}$&\dots&$W_{i-1}$&\dots&$W_{N-2}$&$W_{N-1}$&\dots&$W_{1}$&\dots&$W_{1}$\\
		\hhline{~----------------}
	\end{tabular}
	\caption{Demand set $\{\textbf{\textit{d}}_{l}:1\leq l\leq K\}$}
	\label{demands case 2}
\end{table}

\noindent Consider the user index  $\overline{l}$ defined as
 \begin{equation}
 \label{case 2 ZWN}
 {\overline{l}}=\left\{\begin{aligned}
 {N+1-l}, &\text{ for $1\leq l\leq N$}\\
 {K+N+1-l}, &\text{ for $N+1\leq l\leq K$}
 \end{aligned}\right.
 \end{equation}
 It can be noted that in demand $\textbf{\textit{d}}_{l}$, the user $U_{\overline{l}}$ requires the file $W_{N}$. The following lemma is easy to obtain.
 \begin{lemma}
 	\label{case 2 reduction}
 	Let $\textbf{\textit{A}},\textbf{\textit{B}},\textbf{\textit{C}}\subset \{U_{1},\dots, U_{K}\}$ be such that in demand $\textbf{\textit{d}}_{l}$, every user in $\textbf{\textit{B}}$ requests the file $W_{1}$ and users in $\textbf{\textit{C}}$ together request all the files in $\{W_{2}, \dots,  W_{N}\}$. We have
 	\begin{equation*}
 	H(W_{[N-1]},Z_{\textbf{\textit{A}}},X_{\textbf{\textit{d}}_{l}})+\sum_{i\in \textbf{\textit{B}}}H(Z_{i})+\mid \textbf{\textit{B}}\mid H(X_{\textbf{\textit{d}}_{l}})+\mid \textbf{\textit{B}}\mid H(Z_{\textbf{\textit{C}}})\geq H(W_{[N-1]},Z_{\textbf{\textit{A}}\cup \textbf{\textit{B}}},X_{\textbf{\textit{d}}_{l}})+\mid \textbf{\textit{B}}\mid N
 	\end{equation*}
 \end{lemma}
 \begin{proof}
 	We have,
 	\begin{align*}
	H(W_{[N-1]},Z_{\textbf{\textit{A}}},&X_{\textbf{\textit{d}}_{l}})+\sum_{i\in \textbf{\textit{B}}}H(Z_{i})+\mid \textbf{\textit{B}}\mid H(X_{\textbf{\textit{d}}_{l}})+\mid \textbf{\textit{B}}\mid H(Z_{\textbf{\textit{C}}})\\
	=&H(W_{[N-1]},Z_{\textbf{\textit{A}}},X_{\textbf{\textit{d}}_{l}})+\sum_{i\in \textbf{\textit{B}}}\Big[H(Z_{i})+H(X_{\textbf{\textit{d}}_{l}})\Big]+\mid \textbf{\textit{B}}\mid H(Z_{\textbf{\textit{C}}})\\
	\overset{(a)}{\geq}&H(W_{[N-1]},Z_{\textbf{\textit{A}}},X_{\textbf{\textit{d}}_{l}})+\sum_{i\in \textbf{\textit{B}}}H(Z_{i},X_{\textbf{\textit{d}}_{l}})+\mid \textbf{\textit{B}}\mid H(Z_{\textbf{\textit{C}}})\\
	\overset{(b)}{=}&H(W_{[N-1]},Z_{\textbf{\textit{A}}},X_{\textbf{\textit{d}}_{l}})+\sum_{i\in \textbf{\textit{B}}}H(W_{1},Z_{i},X_{\textbf{\textit{d}}_{l}})+\mid \textbf{\textit{B}}\mid H(Z_{\textbf{\textit{C}}})\\
	\overset{(a)}{\geq}&H(W_{[N-1]},Z_{\textbf{\textit{A}}},X_{\textbf{\textit{d}}_{l}})+\Big[H(W_{1},Z_{\textbf{\textit{B}}},X_{\textbf{\textit{d}}_{l}})+(\mid \textbf{\textit{B}}\mid-1)H(W_{1},X_{\textbf{\textit{d}}_{l}})\Big]+\mid \textbf{\textit{B}}\mid H(Z_{\textbf{\textit{C}}})\\
	\overset{(a)}{\geq}&H(W_{[N-1]},Z_{\textbf{\textit{A}}\cup \textbf{\textit{B}}},X_{\textbf{\textit{d}}_{l}})+H(W_{1},Z_{\textbf{\textit{A}}\cap \textbf{\textit{B}}},X_{\textbf{\textit{d}}_{l}})+(\mid \textbf{\textit{B}}\mid-1) H(W_{1},X_{\textbf{\textit{d}}_{l}})+\mid \textbf{\textit{B}}\mid H(Z_{\textbf{\textit{C}}})\\
	\overset{}{\geq}&H(W_{[N-1]},Z_{\textbf{\textit{A}}\cup \textbf{\textit{B}}},X_{\textbf{\textit{d}}_{l}})+\mid \textbf{\textit{B}}\mid \Big[H(W_{1},X_{\textbf{\textit{d}}_{l}})+H(Z_{\textbf{\textit{C}}})\Big]\\
	\overset{(a)}{\geq}&H(W_{[N-1]},Z_{\textbf{\textit{A}}\cup \textbf{\textit{B}}},X_{\textbf{\textit{d}}_{l}})+\mid \textbf{\textit{B}}\mid H(W_{1},Z_{\textbf{\textit{C}}},X_{\textbf{\textit{d}}_{l}})\\
	\overset{(c)}{=}&H(W_{[N-1]},Z_{\textbf{\textit{A}}\cup \textbf{\textit{B}}},X_{\textbf{\textit{d}}_{l}})+\mid \textbf{\textit{B}}\mid H(W_{[N]},Z_{\textbf{\textit{C}}},X_{\textbf{\textit{d}}_{l}})\\
	\overset{(d)}{=}&H(W_{[N-1]},Z_{\textbf{\textit{A}}\cup \textbf{\textit{B}}},X_{\textbf{\textit{d}}_{l}})+\mid \textbf{\textit{B}}\mid H(W_{[N]})\\\geq&H(W_{[N-1]},Z_{\textbf{\textit{A}}\cup \textbf{\textit{B}}},X_{\textbf{\textit{d}}_{l}})+\mid \textbf{\textit{B}}\mid N
 	\end{align*}
 	where
 	
 	\begin{tabular}{cl}
 		$(a)$& follows form the submodularity property of entropy,\\
 		$(b)$& follows from (\ref{I.1}) and the definition of $\textbf{\textit{B}}$,\\
 		$(c)$& follows from (\ref{I.1}) and the definition of $\textbf{\textit{C}}$,\\
 		$(d)$& follows from (\ref{I.2}).
 	\end{tabular}
 
 \end{proof}

 For $1\leq i\leq N$, consider the sets of users as shown below:
 \begin{table}[!h]
 	\centering
 	\begin{tabular}{|c|c|c|c|}
 		\hline
 		Set& Users& Number& Files Requested in Demand $\textbf{\textit{d}}_{i}$\\
 		\hline
 		 ${\textbf{\textit{A}}_{i}}$&$U_{1},\dots, U_{N-i}$&$N-i$&$W_{i},\dots,W_{N-1}$\\
 		 \hline
 		 ${\textbf{\textit{B}}_{i}}$&$U_{K-i+2},\dots, U_{K}$&$i-1$&$W_{1}\dots, W_{i-1}$\\
 		 \hline
 		 ${\textbf{\textit{F}}_{i}}$&$U_{N+1},\dots, U_{2N-i}$&$N-i$&$W_{i},\dots,W_{N-1}$\\
 		 \hline
 		 ${\textbf{\textit{G}}_{i}}$&$U_{2N-i+1},\dots, U_{K-i+1}$&$K-2N+1$&$W_{1}$\\
 		 \hline
 		 $\textbf{\textit{J}}_{i}$&$U_{1},\dots,U_{N-i+1}$&$N-i+1$&$W_{i}\dots, W_{N}$\\
 		 \hline
 		 $\textbf{\textit{K}}_{i}$&$U_{\substack{K-i+3}}\dots U_{K}$&$i-2$&$W_{2}\dots, W_{i-1}$\\
 		 \hline
 	\end{tabular}
 \end{table}
 
 \noindent These sets are also indicated in TABLE \ref{demands case 2}. Let 
 \begin{align}
 \label{case 2 I}
 \textbf{\textit{I}}_{i}=&\textbf{\textit{J}}_{i}\cup \textbf{\textit{K}}_{i}\\
 \label{case 2 L}
 \textbf{\textit{L}}_{i}=&\textbf{\textit{A}}_{i}\cup \textbf{\textit{B}}_{i} \cup \textbf{\textit{F}}_{i} \cup \textbf{\textit{G}}_{i}
 \end{align}
 Note that
 \begin{IEEEeqnarray}{rl}
 	\label{case 2 null}
 	\textbf{\textit{A}}_{N}=\textbf{\textit{B}}_{1}=\textbf{\textit{F}}_{N}=\textbf{\textit{K}}_{1}=\textbf{\textit{K}}_{2}&=\phi\\
 	\label{case 2 L increment}
 	\textbf{\textit{L}}_{i+1}\cup\{U_{\overline{i+1}}\}&=\textbf{\textit{L}}_{i}\\
 	\label{case 2 B decrement}
 	\textbf{\textit{B}}_{i}\cup \{U_{\overline{N+i}}\}&=\textbf{\textit{B}}_{i+1}
 \end{IEEEeqnarray}

 It can  be noted that in the demands $\textbf{\textit{d}}_{i}$ and $ \textbf{\textit{d}}_{N+i}$, users in the set $\textbf{\textit{B}}_{i}$ are requesting for the same set of files $\{W_{1},\dots, W_{i-1}\}$ (for $1\leq i\leq N$). Thus, from (\ref{symmetry}) we have
 \begin{equation}
 \label{case 2 symmetry}
 H(W_{[i-1]},Z_{\textbf{\textit{B}}_{i}},X_{\textbf{\textit{d}}_{i}})=H(W_{[i-1]},Z_{\textbf{\textit{B}}_{i}},X_{\textbf{\textit{d}}_{N+i}})
 \end{equation}
 Note that $\mid \textbf{\textit{A}}_{i}\cup \textbf{\textit{B}}_{i}\mid=\mid \textbf{\textit{B}}_{i}\cup\textbf{\textit{F}}_{i}\mid=N-1$. Thus,  we have
 \begin{align}
 \label{case2:A,B}
 (N-1)M+R\geq H(Z_{\textbf{\textit{A}}_{i}\cup \textbf{\textit{B}}_{i}})+H(X_{\textbf{\textit{d}}_{i}})\geq H(Z_{\textbf{\textit{A}}_{i}\cup \textbf{\textit{B}}_{i}},X_{\textbf{\textit{d}}_{i}})
 \end{align}
 Similarly,
 \begin{align}
 \label{case2:B,E}
 (N-1)M+R\geq H(Z_{\textbf{\textit{B}}_{i}\cup \textbf{\textit{F}}_{i}})+H(X_{\textbf{\textit{d}}_{i}})\geq H(Z_{\textbf{\textit{B}}_{i}\cup \textbf{\textit{F}}_{i}},X_{\textbf{\textit{d}}_{i}})
 \end{align}
 \noindent We can now obtain the following lemma:
 \begin{lemma}
 	\label{Lemma casse 2 part 1}
 	The sets $\textbf{\textit{B}}_{i}$ and $\textbf{\textit{L}}_{i}$, defined as above, satisfy
 	\begin{align*}
 	(N^{2}(K-2N+3)&-3N+1)M+(N(K-2N+3)-1)R\\ \geq& H(W_{[N-1]},Z_{\textbf{\textit{L}}_{N}})+\sum_{i=1}^{N-1}H(W_{[N-1]},Z_{\textbf{\textit{B}}_{i}},X_{\textbf{\textit{d}}_{N+i}})+N((K-2N+2)N-1)
 	\end{align*} 
 \end{lemma}
 
 \begin{proof}
 	 We have,
 	\begin{IEEEeqnarray*}{rl}
 	(N^{2}&(K-2N+3)-3N+1)M+(N(K-2N+3)-1)R\\
 	=&\sum_{i=1}^{N}\Big[(N-1)M+R+(K-2N+1)M+(K-2N+1)(R+(N-1)M)\Big]+\sum_{i=1}^{N-1}[(N-1)M+R]\\
 	\overset{(a)}{\geq}&\sum_{i=1}^{N}\bigg[H(Z_{\textbf{\textit{A}}_{i}\cup\textbf{\textit{B}}_{i}},X_{\textbf{\textit{d}}_{i}})+\sum_{j\in \textbf{\textit{G}}_{i}}H(Z_{j})+\mid \textbf{\textit{G}}_{i} \mid H(X_{\textbf{\textit{d}}_{i}})+\mid \textbf{\textit{G}}_{i} \mid H(Z_{\textbf{\textit{I}}_{i}})\bigg]+\sum_{i=1}^{N-1}H(Z_{\textbf{\textit{B}}_{i}\cup\textbf{\textit{F}}_{i}},X_{\textbf{\textit{d}}_{i}})\\
 	\overset{(b)}{=}&\sum_{i=1}^{N}\bigg[H(W_{[N-1]},Z_{\textbf{\textit{A}}_{i}\cup\textbf{\textit{B}}_{i}},X_{\textbf{\textit{d}}_{i}})+\sum_{j\in \textbf{\textit{G}}_{i}}H(Z_{j})+\mid \textbf{\textit{G}}_{i} \mid H(X_{\textbf{\textit{d}}_{i}})+\mid \textbf{\textit{G}}_{i} \mid H(Z_{\textbf{\textit{I}}_{i}})\bigg]\\&+\sum_{i=1}^{N-1}H(W_{[N-1]},Z_{\textbf{\textit{B}}_{i}\cup\textbf{\textit{F}}_{i}},X_{\textbf{\textit{d}}_{i}})\\
 	\overset{(c)}{\geq}&\sum_{i=1}^{N}\bigg[H(W_{[N-1]},Z_{\textbf{\textit{A}}_{i}\cup\textbf{\textit{B}}_{i}\cup \textbf{\textit{G}}_{i}},X_{\textbf{\textit{d}}_{i}})+\mid \textbf{\textit{G}}_{i} \mid N\bigg]+\sum_{i=1}^{N-1}H(W_{[N-1]},Z_{\textbf{\textit{B}}_{i}\cup\textbf{\textit{F}}_{i}},X_{\textbf{\textit{d}}_{i}})\\
 	\overset{}{=}&H(W_{[N-1]},Z_{\textbf{\textit{A}}_{N}\cup\textbf{\textit{B}}_{N}\cup \textbf{\textit{G}}_{N}},X_{\textbf{\textit{d}}_{N}})+\sum_{i=1}^{N-1}\Big[H(W_{[N-1]},Z_{\textbf{\textit{A}}_{i}\cup\textbf{\textit{B}}_{i}\cup \textbf{\textit{G}}_{i}},X_{\textbf{\textit{d}}_{i}})+H(W_{[N-1]},Z_{\textbf{\textit{B}}_{i}\cup\textbf{\textit{F}}_{i}},X_{\textbf{\textit{d}}_{i}})\Big]\\&+\sum_{i=1}^{N}\mid \textbf{\textit{G}}_{i} \mid N\\
 	\overset{(d)}{\geq}&H(W_{[N-1]},Z_{\textbf{\textit{A}}_{N}\cup\textbf{\textit{B}}_{N}\cup \textbf{\textit{G}}_{N}\cup\textbf{\textit{F}}_{N}},X_{\textbf{\textit{d}}_{N}})+\sum_{i=1}^{N-1}\Big[H(W_{[N-1]},Z_{\textbf{\textit{A}}_{i}\cup\textbf{\textit{B}}_{i}\cup \textbf{\textit{G}}_{i}\cup\textbf{\textit{F}}_{i}},X_{\textbf{\textit{d}}_{i}})+H(W_{[N-1]},Z_{\textbf{\textit{B}}_{i}},X_{\textbf{\textit{d}}_{i}})\Big]\\&+(K-2N+1) N^{2}\\
 	\overset{(e)}{=}&\sum_{i=1}^{N}H(W_{[N-1]},Z_{\textbf{\textit{L}}_{i}},X_{\textbf{\textit{d}}_{i}})+\sum_{i=1}^{N-1}H(W_{[N-1]},Z_{\textbf{\textit{B}}_{i}},X_{\textbf{\textit{d}}_{i}})+(K-2N+1) N^{2}\\
 	\overset{}{\geq}&\left[H(W_{[N-1]},Z_{\textbf{\textit{L}}_{1}})+\sum_{i=2}^{N}H(W_{[N-1]},Z_{\textbf{\textit{L}}_{i}},X_{\textbf{\textit{d}}_{i}})\right]+\sum_{i=1}^{N-1}H(W_{[N-1]},Z_{\textbf{\textit{B}}_{i}},X_{\textbf{\textit{d}}_{i}})+(K-2N+1) N^{2}\\
 	\overset{(f)}{\geq}&\Big[(N-1)N+H(W_{[N-1]},Z_{\textbf{\textit{L}}_{N}})\Big]+\sum_{i=1}^{N-1}H(W_{[N-1]},Z_{\textbf{\textit{B}}_{i}},X_{\textbf{\textit{d}}_{i}})+(K-2N+1) N^{2}\\
 	\overset{(g)}{=}&H(W_{[N-1]},Z_{\textbf{\textit{L}}_{N}})+\sum_{i=1}^{N-1}H(W_{[N-1]},Z_{\textbf{\textit{B}}_{i}},X_{\textbf{\textit{d}}_{N+i}})+N((K-2N+2)N-1)
 	\end{IEEEeqnarray*}
 	where 
 	
 	\begin{tabular}{cl}
 		$(a)$& follows from (\ref{case2:A,B}) and (\ref{case2:B,E}),\\
 		$(b)$& follows from (\ref{I.1}) and definition of set $\textbf{\textit{A}}_{i}$, $\textbf{\textit{B}}_{i}$, and $\textbf{\textit{F}}_{i}$,\\
 		$(c)$& follows from Lemma \ref{case 2 reduction} with $\textbf{\textit{A}}=\textbf{\textit{A}}_{i}\cup \textbf{\textit{B}}_{i}$, $\textbf{\textit{B}}=\textbf{\textit{G}}_{i}$ and $\textbf{\textit{C}}=\textbf{\textit{I}}_{i}$,\\
 		$(d)$& follows from (\ref{case 2 null}) the submodularity property of entropy,\\
 		$(e)$& follows from the definition of $\textbf{\textit{L}}_{i}$,\\
 		$(f)$& follows from Lemma \ref{Lemma sum increment} with $\textbf{\textit{S}}_{i}=\textbf{\textit{L}}_{i}$, $l=1$, $j=N$ and (\ref{case 2 L increment}),\\
 		$(g)$& follows from (\ref{case 2 symmetry}).
 	\end{tabular}

 \end{proof}

Now, for $2N\leq j\leq K$, consider another sets of users as shown below:

\begin{table}[h]
	\centering
	\begin{tabular}{|c|c|c|c|}
		\hline
		Set & Users & Number& Files Requested in Demand $\textbf{\textit{d}}_{j}$\\
		\hline
		${\textbf{\textit{P}}_{j}}$&$U_{K+N+2-j},\dots, U_{K+2N-j}$&$N-1$& $W_{1},\dots, W_{N-1}$\\
		\hline
		${\textbf{\textit{Q}}_{j}}$&$U_{K+2N+1-j},\dots, U_{K}$&$j-2N$&$W_{1}$\\
		\hline
		$\textbf{\textit{S}}_{j}$&$U_{K-j+3},\dots,U_{K+N-j+2}$&$N-1$&$W_{2},\dots, W_{N}$\\
		\hline
	\end{tabular}
\end{table}
These sets are also indicated in TABLE \ref{demands case 2}. Let
\begin{equation}
\label{case 2 T}
\textbf{\textit{T}}_{j}=\textbf{\textit{P}}_{j}\cup \textbf{\textit{Q}}_{j}
\end{equation}
Note that
\begin{IEEEeqnarray}{rl}
\label{case 2 null Q}
\textbf{\textit{Q}}_{2N}=\textbf{\textit{S}}_{2N}&=\phi\\
\label{case 2 T decrement}
\textbf{\textit{T}}_{j+1}\cup \{U_{\overline{j+1}}\}&=\textbf{\textit{T}}_{j}\\
\label{case 2 LT}
\textbf{\textit{T}}_{K}\cup \{U_{\overline{K}}\}&=\textbf{\textit{L}}_{N}\\
\label{case 2 BT}
\textbf{\textit{B}}_{N-1}\cup \{U_{\overline{2N-1}}\}=\textbf{\textit{B}}_{N}&=\textbf{\textit{T}}_{2N}
\end{IEEEeqnarray}
Note that $\mid \textbf{\textit{P}}_j\mid= N-1$. Thus, we have
\begin{align}
\label{case2:P}
(N-1)M+R\geq H(Z_{\textbf{\textit{P}}_{j}})+H(X_{\textbf{\textit{d}}_j})\geq H(Z_{\textbf{\textit{P}}_{j}},X_{\textbf{\textit{d}}_{j}})
\end{align}
The following lemma is easy to obtain:
\begin{lemma}
	\label{Lemma case 2 part 2}
	The set $\textbf{\textit{T}}_{j}$, as defined above, satisfy
	\begin{IEEEeqnarray*}{rl}
	\frac{(K-2N+1)}{2}\bigg[\big(N(K-2N+2)-2\big) M+(K-2N+2)R\bigg] \geq & \sum_{j=2N}^{K}H(W_{[N-1]},Z_{\textbf{\textit{T}}_{j}},X_{\textbf{\textit{d}}_{j}})\\&+\frac{(K-2N+1)(K-2N)}{2}N
	\end{IEEEeqnarray*}
\end{lemma}
\begin{proof}
	We have,
	\begin{IEEEeqnarray*}{l}
		\frac{(K-2N+1)}{2}\bigg[\big(N(K-2N+2)-2\big) M+(K-2N+2)R\bigg]\\
		=\sum_{j=2N}^{K}\Big[((N-1)M+R)+(j-2N)M+(j-2N)(R+(N-1)M)\Big]\\
		\overset{(a)}{\geq}  \sum_{j=2N}^{K}\left[H(Z_{\textbf{\textit{P}}_{j}},X_{\textbf{\textit{d}}_{j}})+ \sum_{l\in \textbf{\textit{Q}}_{j}}H(Z_{l})+\mid \textbf{\textit{Q}}_{j}\mid H(X_{\textbf{\textit{d}}_{j}})+\mid \textbf{\textit{Q}}_{j} \mid H(Z_{\textbf{\textit{S}}_{j}})\right]\\
		\overset{(b)}{=}  \sum_{j=2N}^{K}\left[H(W_{[N-1]},Z_{\textbf{\textit{P}}_{j}},X_{\textbf{\textit{d}}_{j}})+ \sum_{l\in \textbf{\textit{Q}}_{j}}H(Z_{l})+\mid \textbf{\textit{Q}}_{j} \mid H(X_{\textbf{\textit{d}}_{j}})+\mid \textbf{\textit{Q}}_{j} \mid H(Z_{\textbf{\textit{S}}_{j}})\right]\\
		\overset{(c)}{\geq}  \sum_{j=2N}^{K}\left[H(W_{[N-1]},Z_{\textbf{\textit{P}}_{j}\cup \textbf{\textit{Q}}_{j}},X_{\textbf{\textit{d}}_{j}})+\mid \textbf{\textit{Q}}_{j} \mid N\right]\\
		\overset{(d)}{=}\sum_{j=2N}^{K}H(W_{[N-1]},Z_{\textbf{\textit{T}}_{j}},X_{\textbf{\textit{d}}_{j}})+\sum_{j=2N}^{K}(j-2N)N\\
		\overset{}{=}\sum_{j=2N}^{K}H(W_{[N-1]},Z_{\textbf{\textit{T}}_{j}},X_{\textbf{\textit{d}}_{j}})+\frac{(K-2N+1)(K-2N)}{2}N
	\end{IEEEeqnarray*}
where 

\begin{tabular}{cl}
	$(a)$& follows from (\ref{case2:P}) and the fact that $\textbf{\textit{Q}}_{2N}=\phi$,\\
	$(b)$& follows from definition of sets $\textbf{\textit{P}}_{j}$, $\textbf{\textit{Q}}_{j}$ and (\ref{I.1}),\\
	$(c)$& follows from Lemma \ref{case 2 reduction} with $\textbf{\textit{A}}=\textbf{\textit{P}}_{j}$, $\textbf{\textit{B}}=\textbf{\textit{Q}}_{j}$ and $\textbf{\textit{C}}=\textbf{\textit{S}}_{j}$,\\
	$(d)$& follows from the definition of $\textbf{\textit{T}}_{j}$.
\end{tabular}

\end{proof}
\noindent Using the above lemma, we can obtain the following result:
\begin{theorem}
	\label{Theorem:genera_case_2}
	For the $(N,K)$ cache network, when $1\leq N< \big\lceil\frac{K+1}{2}\big\rceil$, achievable memory rate pairs $(M, R)$ must satisfy the constraint
	\begin{align*}	
	\frac{K(N(K+3)-2(N^{2}+1))}{2}&M+\frac{K(K+3-2N)}{2}R\geq \frac{NK(K-2N+3)-2}{2}
	\end{align*}
\end{theorem}

\begin{proof}
	We have,
	\begin{IEEEeqnarray*}{l}
	\frac{K(N(K+3)-2(N^{2}+1))}{2} M+\frac{K(K+3-2N)}{2}R\\
		= \Big[(N^{2}(K-2N+3)-3N+1)M+(N(K-2N+3)-1)R\Big] \\+\left[\frac{(K-2N+1)}{2}\bigg(\big(N(K-2N+2)-2\big) M+(K-2N+2)R\bigg)\right] \\
		\overset{(a)}{\geq}  \left[H(W_{[N-1]},Z_{\textbf{\textit{L}}_{N}})+\sum_{i=1}^{N-1}H(W_{[N-1]},Z_{\textbf{\textit{B}}_{i}},X_{\textbf{\textit{d}}_{N+i}})+N((K-2N+2)N-1)\right]\\+\left[\sum_{j=2N}^{K}H(W_{[N-1]},Z_{\textbf{\textit{T}}_{j}},X_{\textbf{\textit{d}}_{j}})+\frac{(K-2N+1)(K-2N)}{2}N\right]\\
		\overset{(b)}{=} \left[H(W_{[N-1]},Z_{\textbf{\textit{T}}_{K}},Z_{\overline{K}})+ \sum_{j=2N}^{K}H(W_{[N-1]},Z_{\textbf{\textit{T}}_{j}},X_{\textbf{\textit{d}}_{j}})\right]+\sum_{i=1}^{N-1}H(W_{[N-1]},Z_{\textbf{\textit{B}}_{i}},X_{\textbf{\textit{d}}_{(N+i)}})\\+\frac{(K(K-2N+1)+2N-2)}{2}N\\
		\overset{(c)}{\geq} \Big[(K-2N+1)N+H(W_{[N-1]},Z_{\textbf{\textit{T}}_{2N}})\Big] +\sum_{i=1}^{N-1}H(W_{[N-1]},Z_{\textbf{\textit{B}}_{i}},X_{\textbf{\textit{d}}_{(N+i)}})+\frac{(K(K-2N+1)+2N-2)}{2}N\\
		\overset{(d)}{=}\left[H(W_{[N-1]},Z_{\textbf{\textit{B}}_{N-1}},Z_{\overline{2N-1}}) +\!\sum_{i=1}^{N-1}\!H(W_{[N-1]},Z_{\textbf{\textit{B}}_{i}},X_{\textbf{\textit{d}}_{(N+i)}})\right]+\frac{(K(K-2N+3)-2N)}{2}N\\
		\overset{(e)}{\geq} \Big[(N-1)N+H(W_{[N-1]},Z_{\textbf{\textit{B}}_{1}})\Big]+\frac{(K(K-2N+3)-2N)}{2}N\\
		\overset{(f)}{=} H(W_{[N-1]})+\frac{(K(K-2N+3)-2)}{2}N\\
		\geq\frac{NK(K-2N+3)-2}{2}
	\end{IEEEeqnarray*}
	where 
	
	\begin{tabular}{cl}
		$(a)$& follows from Lemma \ref{Lemma casse 2 part 1} and Lemma \ref{Lemma case 2 part 2},\\
		$(b)$& follows from (\ref{case 2 LT}),\\
		$(c)$& follows from (\ref{Lemma sum decrement}) with $\textbf{\textit{T}}_{i}=\textbf{\textit{T}}_{j}$, $l=2N$, $j=K$  and (\ref{case 2 T decrement}),\\
		$(d)$& follows from (\ref{case 2 BT}),\\
		$(e)$& follows from (\ref{Lemma sum decrement}) with $\textbf{\textit{T}}_{i}=\textbf{\textit{B}}_{i}$, $l=1$, $j=N-1$ and (\ref{case 2 B decrement}),\\
		$(f)$& follows from (\ref{case 2 null}).
	\end{tabular}

\end{proof}

\section{Comparison with previous bounds}

In \cite{maddah2014fundamental}, Maddah-Ali and Niesen derived a lower bound on achievable rates using cut set arguments, which was further improved in \cite{ajaykrishnan2015critical,ghasemi2017improved,sengupta2015improved,wang2018improved,yu2017characterizing}. A comparison between these lower bounds and the new lower bounds in Section III, at cache size $M=\frac{N}{K(N-1)}$, is given in TABLE \ref{TABLE:Rate comparison}. 
\begin{table}[h]
	\centering
		\setlength{\tabcolsep}{4.5pt}
	\begin{tabular}{|c|c|c|c|}
		\hline
		Lower Bound & Case I: $\big\lceil\frac{K+1}{2}\big\rceil \leq N\leq K$ & Case II: $1\leq N \leq\big\lceil\frac{K+1}{2}\big\rceil$   \\[0.4em]
		\hline&&\\[-1.5em]
		{Cut Set bound \cite{maddah2014fundamental}}& $N-\frac{N^{2}}{(N-1)K}$ &$N-\frac{N^{2}}{(N-1)K}$  \\ 
		\hline&&\\[-1.5em]
		$\begin{array}{c}
		\text{Ghasemi \&}\\ \text{ Ramamoorthy \cite{ghasemi2017improved}}
		\end{array}$
		& $N-\frac{N^{2}}{(N-1)K}$&$N-\frac{N^{2}}{(N-1)K}$   \\ [0.4em]
		\hline&&\\[-1.5em]
		Ajaykrishnan et al. \cite{ajaykrishnan2015critical}& $N-\frac{N^{2}}{(N-1)K}$&$N-\frac{N^{2}}{(N-1)K}$   \\ [0.4em]
		\hline&&\\[-1.5em]
		{Wang et al. \cite{wang2018improved}}& $N-\frac{N^{2}}{(N-1)K}$&$N-\frac{N^{2}}{(N-1)K}$  \\[0.4em] 
		\hline&&\\[-1.5em]
		Yu et al. \cite{yu2017characterizing}&$N-\frac{N^{2}}{(N-1)K}+\frac{1}{K(N-1)}\left(N-K+\frac{K}{N}\right)$&$N-\frac{N^{2}}{(N-1)K}$  \\ [0.4em]
		\hline&&\\[-1.5em]
		Sengupta et al. \cite{sengupta2015improved}& $N-\frac{N^{2}}{(N-1)K}+\frac{1}{K(N-1)}\left(N-K+\frac{K}{N}\right)$&$N-\frac{N^{2}}{(N-1)K}$  \\ [0.4em]
		\hline&&\\[-1.5em]
		New lower bound
		& $N-\frac{N^{2}}{(N-1)K}+\frac{1}{K(N-1)}$ &$\begin{array}{c}
		N-\frac{N^{2}}{K(N-1)}+\frac{2}{K(N-1)(K+3-2N)}
		\end{array}$ \\[0.4em]
		\hline
	\end{tabular} 
	
	\caption{Comparison with previous lower bounds for $M=\frac{N}{K(N-1)}$}
	\label{TABLE:Rate comparison}
\end{table}
It can be noted that the new bounds improve upon the previous ones. For the $(N,K)$ cache network, the scheme proposed by  G{\'o}mez-Vilardeb{\'o} in \cite{gomez2018fundamental} achieves memory rate pairs $(M,R_{G})=\left(M,\dfrac{KN-1}{K}-(N-1)M\right)$, for $M\in \big[\frac{1}{K},\frac{N}{(N-1)K}\big]$. From Theorem 1, when $\big\lceil\frac{K+1}{2}\big\rceil\leq N\leq K$, we have that all achievable memory rate pairs satisfy the constraint
\begin{align*}
R\geq &\dfrac{KN-1}{K}-(N-1)M
=R_{G}(M)
\end{align*}
Thus we have:
\begin{theorem}
	\label{(N,K)result}
	For the $(N,K)$ cache network, when $\big\lceil\frac{K+1}{2}\big\rceil\leq N\leq K$,  
	the exact rate memory tradeoff is given by
	\begin{equation}
	R^{*}(M)=\dfrac{KN-1}{K}-(N-1)M
	\end{equation}
	where $M\in \left[\frac{1}{K},\frac{N}{K(N-1)}\right]$.
\end{theorem}


\begin{remark}
	In \cite{gomez2018fundamental}, with the help of  the lower bounds derived in \cite{sengupta2015improved} and \cite{yu2017characterizing}, G{\'o}mez-Vilardeb{\'o} showed that when $K=N$ and $M=\frac{1}{N-1}$, his scheme is optimal. We extend his result to the case where $\big\lceil\frac{K+1}{2}\big\rceil \leq N \leq K$ in Theorem \ref{(N,K)result}.
\end{remark}

\section{Conclusions}
In this paper we considered the canonical $(N,K)$ cache network where $N \leq K$ and $M\in\big[0,\frac{N}{K}\big]$. We derived a new set of lower bounds on the achievable rate when each file in the server is requested by at least one user. Using these lower bounds, we showed that when $\big\lceil\frac{K+1}{2}\big\rceil \leq N\leq K$ the scheme proposed in \cite{gomez2018fundamental} is optimal for $M\in\left[\frac{1}{K},\frac{N}{K(N-1)}\right]$. For the case $1\leq N\leq \big\lceil\frac{K+1}{2}\big\rceil$, the new lower bound was shown to improve upon the previous lower bounds, but a matching scheme is still not known. The work presented forms another step in the attempt to find a characterization  of the exact rate memory tradeoff for coded caching and is illustrated in Fig. \ref{(N,K)bound plot}.

\begin{figure}[h]
	\centering
	\input{NKbound}
	\caption{Rate memory tradeoff for the $(N,K)$ cache network  when $\big\lceil\frac{K+1}{2}\big\rceil< N\leq K$}
	\label{(N,K)bound plot}
\end{figure}
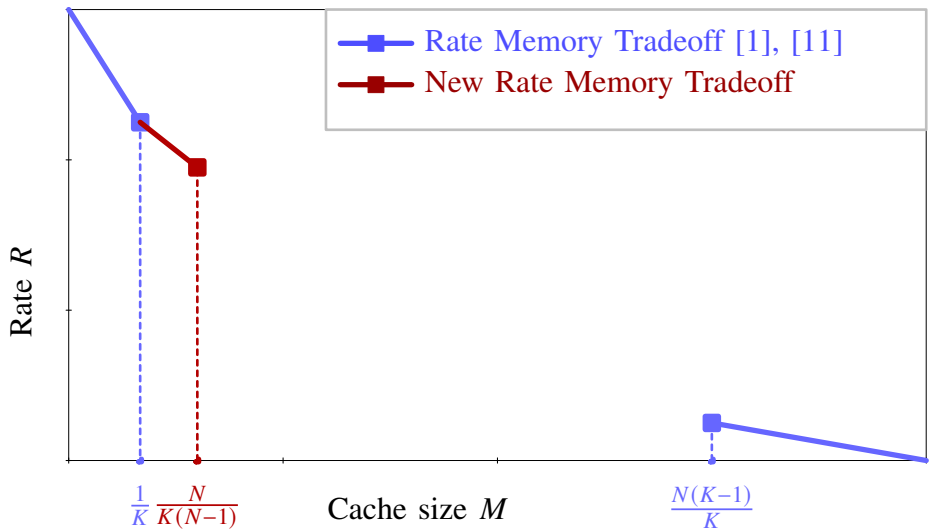


\end{document}

%% file: can.tex
\ifx\du\undefined
  \newlength{\du}
\fi
\setlength{\du}{10\unitlength}
\begin{tikzpicture}
\pgftransformxscale{1.000000}
\pgftransformyscale{-1.000000}
\definecolor{dialinecolor}{rgb}{0.000000, 0.000000, 0.000000}
\pgfsetstrokecolor{dialinecolor}
\definecolor{dialinecolor}{rgb}{1.000000, 1.000000, 1.000000}
\pgfsetfillcolor{dialinecolor}
\pgfsetlinewidth{0.100000\du}
\pgfsetdash{}{0pt}
\pgfsetdash{}{0pt}
\pgfsetmiterjoin
\definecolor{dialinecolor}{rgb}{1.000000, 1.000000, 1.000000}
\pgfsetfillcolor{dialinecolor}
\fill (17.5\du,11.096700\du)--(17.5\du,13.746700\du)--(20.8\du,13.746700\du)--(20.8\du,11.096700\du)--cycle;
\definecolor{dialinecolor}{rgb}{0.000000, 0.000000, 0.000000}
\pgfsetstrokecolor{dialinecolor}
\draw (17.5\du,11.096700\du)--(17.5\du,13.746700\du)--(20.8\du,13.746700\du)--(20.8\du,11.096700\du)--cycle;
\pgfsetlinewidth{0.100000\du}
\pgfsetdash{}{0pt}
\pgfsetdash{}{0pt}
\pgfsetmiterjoin
\definecolor{dialinecolor}{rgb}{1.000000, 1.000000, 1.000000}
\pgfsetfillcolor{dialinecolor}
\fill (20.8\du,11.096700\du)--(20.8\du,13.746700\du)--(24\du,13.746700\du)--(24\du,11.096700\du)--cycle;
\definecolor{dialinecolor}{rgb}{0.000000, 0.000000, 0.000000}
\pgfsetstrokecolor{dialinecolor}
\draw (20.8\du,11.096700\du)--(20.8\du,13.746700\du)--(24\du,13.746700\du)--(24\du,11.096700\du)--cycle;
\pgfsetlinewidth{0.100000\du}
\pgfsetdash{}{0pt}
\pgfsetdash{}{0pt}
\pgfsetmiterjoin
\definecolor{dialinecolor}{rgb}{1.000000, 1.000000, 1.000000}
\pgfsetfillcolor{dialinecolor}
\fill (24.\du,11.096700\du)--(24.\du,13.746700\du)--(30.3\du,13.746700\du)--(30.3\du,11.096700\du)--cycle;
\definecolor{dialinecolor}{rgb}{0.000000, 0.000000, 0.000000}
\pgfsetstrokecolor{dialinecolor}
\draw (24.\du,11.096700\du)--(24.\du,13.746700\du)--(30.1\du,13.746700\du)--(30.1\du,11.096700\du)--cycle;
\pgfsetlinewidth{0.100000\du}
\pgfsetdash{}{0pt}
\pgfsetdash{}{0pt}
\pgfsetmiterjoin
\definecolor{dialinecolor}{rgb}{1.000000, 1.000000, 1.000000}
\pgfsetfillcolor{dialinecolor}
\fill (29.6\du,11.096700\du)--(29.6\du,13.746700\du)--(33.3\du,13.746700\du)--(33.3\du,11.096700\du)--cycle;
\definecolor{dialinecolor}{rgb}{0.000000, 0.000000, 0.000000}
\pgfsetstrokecolor{dialinecolor}
\draw (29.6\du,11.096700\du)--(29.6\du,13.746700\du)--(33.3\du,13.746700\du)--(33.3\du,11.096700\du)--cycle;
\pgfsetlinewidth{0.100000\du}
\pgfsetdash{}{0pt}
\pgfsetdash{}{0pt}
\pgfsetmiterjoin
\definecolor{dialinecolor}{rgb}{1.000000, 1.000000, 1.000000}
\pgfsetfillcolor{dialinecolor}
\fill (33.3\du,11.096700\du)--(33.3\du,13.746700\du)--(36.5\du,13.746700\du)--(36.5\du,11.096700\du)--cycle;
\definecolor{dialinecolor}{rgb}{0.000000, 0.000000, 0.000000}
\pgfsetstrokecolor{dialinecolor}
\draw (33.3\du,11.096700\du)--(33.3\du,13.746700\du)--(36.5\du,13.746700\du)--(36.5\du,11.096700\du)--cycle;
\definecolor{dialinecolor}{rgb}{0.000000, 0.000000, 0.000000}
\pgfsetstrokecolor{dialinecolor}
\node[anchor=west] at (18.\du,12.42\du){$W_{1}$};
\definecolor{dialinecolor}{rgb}{0.000000, 0.000000, 0.000000}
\pgfsetstrokecolor{dialinecolor}
\node[anchor=west] at (21.3\du,12.42\du){$W_{2}$};
\definecolor{dialinecolor}{rgb}{0.000000, 0.000000, 0.000000}
\pgfsetstrokecolor{dialinecolor}
\node[anchor=west] at (26.1\du,12.42\du){$\textbf{\dots}$};
\definecolor{dialinecolor}{rgb}{0.000000, 0.000000, 0.000000}
\pgfsetstrokecolor{dialinecolor}
\node[anchor=west] at (29.3\du,12.42\du){$W_{N-1}$};
\definecolor{dialinecolor}{rgb}{0.000000, 0.000000, 0.000000}
\pgfsetstrokecolor{dialinecolor}
\node[anchor=west] at (33.5\du,12.42\du){$W_{N}$};
\definecolor{dialinecolor}{rgb}{1.000000, 1.000000, 1.000000}
\pgfsetfillcolor{dialinecolor}
\fill (18.020600\du,21.042300\du)--(18.020600\du,22.90\du)--(19.890600\du,22.90\du)--(19.890600\du,21.042300\du)--cycle;
\pgfsetlinewidth{0.100000\du}
\pgfsetdash{}{0pt}
\pgfsetdash{}{0pt}
\pgfsetmiterjoin
\definecolor{dialinecolor}{rgb}{0.000000, 0.000000, 0.000000}
\pgfsetstrokecolor{dialinecolor}
\draw (18.020600\du,21.042300\du)--(18.020600\du,22.90\du)--(19.890600\du,22.90\du)--(19.890600\du,21.042300\du)--cycle;
\definecolor{dialinecolor}{rgb}{0.000000, 0.000000, 0.000000}
\pgfsetstrokecolor{dialinecolor}
\node at (18.955600\du,22.1\du){$U_1$};
\definecolor{dialinecolor}{rgb}{0.000000, 0.000000, 0.000000}
\pgfsetstrokecolor{dialinecolor}
\node[anchor=west] at (18.955600\du,21.992300\du){};
\pgfsetlinewidth{0.100000\du}
\pgfsetdash{}{0pt}
\pgfsetdash{}{0pt}
\pgfsetmiterjoin
\definecolor{dialinecolor}{rgb}{1.000000, 1.000000, 1.000000}
\pgfsetfillcolor{dialinecolor}
\fill (17.413300\du,18.1\du)--(17.413300\du,19.5\du)--(20.478028\du,19.5\du)--(20.478028\du,18.1\du)--cycle;
\definecolor{dialinecolor}{rgb}{0.000000, 0.000000, 0.000000}
\pgfsetstrokecolor{dialinecolor}
\draw (17.413300\du,18.1\du)--(17.413300\du,19.5\du)--(20.478028\du,19.5\du)--(20.478028\du,18.1\du)--cycle;
\definecolor{dialinecolor}{rgb}{0.000000, 0.000000, 0.000000}
\pgfsetstrokecolor{dialinecolor}
\node[anchor=west] at (18.\du,18.87\du){$Z_{1}$};
\pgfsetlinewidth{0.100000\du}
\pgfsetdash{}{0pt}
\pgfsetdash{}{0pt}
\pgfsetbuttcap
{
\definecolor{dialinecolor}{rgb}{0.000000, 0.000000, 0.000000}
\pgfsetfillcolor{dialinecolor}
\pgfsetarrowsend{stealth}
\definecolor{dialinecolor}{rgb}{0.000000, 0.000000, 0.000000}
\pgfsetstrokecolor{dialinecolor}
\draw (18.945600\du,19.5\du)--(18.945600\du,21.042300\du);
}
\definecolor{dialinecolor}{rgb}{1.000000, 1.000000, 1.000000}
\pgfsetfillcolor{dialinecolor}
\fill (22.507400\du,21.052500\du)--(22.507400\du,22.90\du)--(24.377400\du,22.90\du)--(24.377400\du,21.052500\du)--cycle;
\pgfsetlinewidth{0.100000\du}
\pgfsetdash{}{0pt}
\pgfsetdash{}{0pt}
\pgfsetmiterjoin
\definecolor{dialinecolor}{rgb}{0.000000, 0.000000, 0.000000}
\pgfsetstrokecolor{dialinecolor}
\draw (22.507400\du,21.052500\du)--(22.507400\du,22.900\du)--(24.377400\du,22.90\du)--(24.377400\du,21.052500\du)--cycle;
\definecolor{dialinecolor}{rgb}{0.000000, 0.000000, 0.000000}
\pgfsetstrokecolor{dialinecolor}
\node at (23.442400\du,22.1\du){$U_2$};
\definecolor{dialinecolor}{rgb}{0.000000, 0.000000, 0.000000}
\pgfsetstrokecolor{dialinecolor}
\node[anchor=west] at (22.742400\du,22.002500\du){};
\pgfsetlinewidth{0.100000\du}
\pgfsetdash{}{0pt}
\pgfsetdash{}{0pt}
\pgfsetmiterjoin
\definecolor{dialinecolor}{rgb}{1.000000, 1.000000, 1.000000}
\pgfsetfillcolor{dialinecolor}
\fill (21.913800\du,18.1\du)--(21.913800\du,19.5\du)--(24.978528\du,19.5\du)--(24.978528\du,18.1\du)--cycle;
\definecolor{dialinecolor}{rgb}{0.000000, 0.000000, 0.000000}
\pgfsetstrokecolor{dialinecolor}
\draw (21.913800\du,18.1\du)--(21.913800\du,19.5\du)--(24.978528\du,19.5\du)--(24.978528\du,18.1\du)--cycle;
\definecolor{dialinecolor}{rgb}{0.000000, 0.000000, 0.000000}
\pgfsetstrokecolor{dialinecolor}
\node[anchor=west] at (22.50\du,18.87\du){$Z_{2}$};
\pgfsetlinewidth{0.100000\du}
\pgfsetdash{}{0pt}
\pgfsetdash{}{0pt}
\pgfsetbuttcap
{
\definecolor{dialinecolor}{rgb}{0.000000, 0.000000, 0.000000}
\pgfsetfillcolor{dialinecolor}
\pgfsetarrowsend{stealth}
\definecolor{dialinecolor}{rgb}{0.000000, 0.000000, 0.000000}
\pgfsetstrokecolor{dialinecolor}
\draw (23.446100\du,19.5\du)--(23.446100\du,21.052500\du);
}
\definecolor{dialinecolor}{rgb}{1.000000, 1.000000, 1.000000}
\pgfsetfillcolor{dialinecolor}
\fill (34.403350\du,21.051600\du)--(34.403350\du,22.90\du)--(36.300850\du,22.90\du)--(36.300850\du,21.051600\du)--cycle;
\pgfsetlinewidth{0.100000\du}
\pgfsetdash{}{0pt}
\pgfsetdash{}{0pt}
\pgfsetmiterjoin
\definecolor{dialinecolor}{rgb}{0.000000, 0.000000, 0.000000}
\pgfsetstrokecolor{dialinecolor}
\draw (34.403350\du,21.051600\du)--(34.403350\du,22.90\du)--(36.300850\du,22.90\du)--(36.300850\du,21.051600\du)--cycle;
\definecolor{dialinecolor}{rgb}{0.000000, 0.000000, 0.000000}
\pgfsetstrokecolor{dialinecolor}
\node at (35.352100\du,22.1\du){$U_{K}$};
\definecolor{dialinecolor}{rgb}{0.000000, 0.000000, 0.000000}
\pgfsetstrokecolor{dialinecolor}
\node[anchor=west] at (35.352100\du,22.011600\du){};
\pgfsetlinewidth{0.100000\du}
\pgfsetdash{}{0pt}
\pgfsetdash{}{0pt}
\pgfsetmiterjoin
\definecolor{dialinecolor}{rgb}{1.000000, 1.000000, 1.000000}
\pgfsetfillcolor{dialinecolor}
\fill (33.823500\du,18.1\du)--(33.823500\du,19.5\du)--(36.888228\du,19.5\du)--(36.888228\du,18.1\du)--cycle;
\definecolor{dialinecolor}{rgb}{0.000000, 0.000000, 0.000000}
\pgfsetstrokecolor{dialinecolor}
\draw (33.823500\du,18.1\du)--(33.823500\du,19.5\du)--(36.888228\du,19.5\du)--(36.888228\du,18.1\du)--cycle;
\definecolor{dialinecolor}{rgb}{0.000000, 0.000000, 0.000000}
\pgfsetstrokecolor{dialinecolor}
\node[anchor=west] at (34.41\du,18.87\du){$Z_{K}$};
\pgfsetlinewidth{0.100000\du}
\pgfsetdash{}{0pt}
\pgfsetdash{}{0pt}
\pgfsetbuttcap
{
\definecolor{dialinecolor}{rgb}{0.000000, 0.000000, 0.000000}
\pgfsetfillcolor{dialinecolor}
\pgfsetarrowsend{stealth}
\definecolor{dialinecolor}{rgb}{0.000000, 0.000000, 0.000000}
\pgfsetstrokecolor{dialinecolor}
\draw (35.355900\du,19.5\du)--(35.355900\du,21.061600\du);
}
\pgfsetlinewidth{0.100000\du}
\pgfsetdash{}{0pt}
\pgfsetdash{}{0pt}
\pgfsetbuttcap
{
\definecolor{dialinecolor}{rgb}{0.000000, 0.000000, 0.000000}
\pgfsetfillcolor{dialinecolor}
\definecolor{dialinecolor}{rgb}{0.000000, 0.000000, 0.000000}
\pgfsetstrokecolor{dialinecolor}
\draw (27.2\du,13.750000\du)--(27.2\du,16.300000\du);
}
\pgfsetlinewidth{0.100000\du}
\pgfsetdash{}{0pt}
\pgfsetdash{}{0pt}
\pgfsetbuttcap
{
\definecolor{dialinecolor}{rgb}{0.000000, 0.000000, 0.000000}
\pgfsetfillcolor{dialinecolor}
\pgfsetarrowsend{stealth}
\definecolor{dialinecolor}{rgb}{0.000000, 0.000000, 0.000000}
\pgfsetstrokecolor{dialinecolor}
\draw (27.2\du,16.3\du)--(18.945664\du,18.1\du);
}
\pgfsetlinewidth{0.100000\du}
\pgfsetdash{}{0pt}
\pgfsetdash{}{0pt}
\pgfsetbuttcap
{
\definecolor{dialinecolor}{rgb}{0.000000, 0.000000, 0.000000}
\pgfsetfillcolor{dialinecolor}
\pgfsetarrowsend{stealth}
\definecolor{dialinecolor}{rgb}{0.000000, 0.000000, 0.000000}
\pgfsetstrokecolor{dialinecolor}
\draw (27.2\du,16.3\du)--(23.6\du,18.1\du);
}
\pgfsetlinewidth{0.100000\du}
\pgfsetdash{}{0pt}
\pgfsetdash{}{0pt}
\pgfsetbuttcap
{
\definecolor{dialinecolor}{rgb}{0.000000, 0.000000, 0.000000}
\pgfsetfillcolor{dialinecolor}
\pgfsetarrowsend{stealth}
\definecolor{dialinecolor}{rgb}{0.000000, 0.000000, 0.000000}
\pgfsetstrokecolor{dialinecolor}
\draw (27.2\du,16.3\du)--(35.355864\du,18.1\du);
}
\pgfsetlinewidth{0.100000\du}
\pgfsetdash{}{0pt}
\pgfsetdash{}{0pt}
\pgfsetmiterjoin
\definecolor{dialinecolor}{rgb}{1.000000, 1.000000, 1.000000}
\pgfsetfillcolor{dialinecolor}
\fill (25.436281\du,18.393440\du)--(25.436281\du,19.327286\du)--(30.776712\du,19.327286\du)--(30.776712\du,18.393440\du)--cycle;
\definecolor{dialinecolor}{rgb}{1.000000, 1.000000, 1.000000}
\pgfsetstrokecolor{dialinecolor}
\draw (25.436281\du,18.393440\du)--(25.436281\du,19.327286\du)--(30.776712\du,19.327286\du)--(30.776712\du,18.393440\du)--cycle;
\definecolor{dialinecolor}{rgb}{0.000000, 0.000000, 0.000000}
\pgfsetstrokecolor{dialinecolor}
\node[anchor=west] at (26.5\du,18.9\du){$\textbf{.}\ \ \textbf{.} \ \ \textbf{.}$};
\definecolor{dialinecolor}{rgb}{0.000000, 0.000000, 0.000000}
\pgfsetstrokecolor{dialinecolor}
\node[anchor=west] at (23.422676\du,18.787407\du){};
\pgfsetlinewidth{0.100000\du}
\pgfsetdash{}{0pt}
\pgfsetdash{}{0pt}
\pgfsetbuttcap
{
\definecolor{dialinecolor}{rgb}{0.000000, 0.000000, 0.000000}
\pgfsetfillcolor{dialinecolor}
\pgfsetarrowsend{to}
\definecolor{dialinecolor}{rgb}{0.000000, 0.000000, 0.000000}
\pgfsetstrokecolor{dialinecolor}
\draw (20.475225\du,10.59\du)--(36.3\du,10.59\du);
}
\pgfsetlinewidth{0.100000\du}
\pgfsetdash{}{0pt}
\pgfsetdash{}{0pt}
\pgfsetbuttcap
{
\definecolor{dialinecolor}{rgb}{0.000000, 0.000000, 0.000000}
\pgfsetfillcolor{dialinecolor}
\pgfsetarrowsend{to}
\definecolor{dialinecolor}{rgb}{0.000000, 0.000000, 0.000000}
\pgfsetstrokecolor{dialinecolor}
\draw (33.782528\du,10.59\du)--(17.7\du,10.59\du);
}
\pgfsetlinewidth{0.100000\du}
\pgfsetdash{}{0pt}
\pgfsetdash{}{0pt}
\pgfsetbuttcap
{
\definecolor{dialinecolor}{rgb}{0.000000, 0.000000, 0.000000}
\pgfsetfillcolor{dialinecolor}
\pgfsetarrowsend{to}
\definecolor{dialinecolor}{rgb}{0.000000, 0.000000, 0.000000}
\pgfsetstrokecolor{dialinecolor}
\draw (37.45\du,18.700287\du)--(37.45\du,19.49\du);
}
\pgfsetlinewidth{0.100000\du}
\pgfsetdash{}{0pt}
\pgfsetdash{}{0pt}
\pgfsetbuttcap
{
\definecolor{dialinecolor}{rgb}{0.000000, 0.000000, 0.000000}
\pgfsetfillcolor{dialinecolor}
\pgfsetarrowsend{to}
\definecolor{dialinecolor}{rgb}{0.000000, 0.000000, 0.000000}
\pgfsetstrokecolor{dialinecolor}
\draw (37.45\du,19.006705\du)--(37.45\du,18.233364\du);
}
\definecolor{dialinecolor}{rgb}{0.000000, 0.000000, 0.000000}
\pgfsetstrokecolor{dialinecolor}
\node[anchor=west] at (24.8\du,9.8\du){$N$ Files};
\definecolor{dialinecolor}{rgb}{0.000000, 0.000000, 0.000000}
\pgfsetstrokecolor{dialinecolor}
\node[anchor=west] at (12\du,12.053931\du){Server};
\definecolor{dialinecolor}{rgb}{0.000000, 0.000000, 0.000000}
\pgfsetstrokecolor{dialinecolor}
\definecolor{dialinecolor}{rgb}{0.000000, 0.000000, 0.000000}
\pgfsetstrokecolor{dialinecolor}
\node[anchor=west] at (37.5\du,18.817017\du){$M$};
\definecolor{dialinecolor}{rgb}{0.000000, 0.000000, 0.000000}
\pgfsetstrokecolor{dialinecolor}
\node[anchor=west] at (12\du,18.962931\du){Caches};
\definecolor{dialinecolor}{rgb}{0.000000, 0.000000, 0.000000}
\pgfsetstrokecolor{dialinecolor}
\node[anchor=west] at (12\du,22.085478\du){Users};
\definecolor{dialinecolor}{rgb}{0.000000, 0.000000, 0.000000}
\pgfsetstrokecolor{dialinecolor}
\node[anchor=west] at (16.5\du,24.5\du){User 1};
\definecolor{dialinecolor}{rgb}{0.000000, 0.000000, 0.000000}
\pgfsetstrokecolor{dialinecolor}
\node[anchor=west] at (21\du,24.5\du){User 2};
\definecolor{dialinecolor}{rgb}{0.000000, 0.000000, 0.000000}
\pgfsetstrokecolor{dialinecolor}
\node[anchor=west] at (33.0\du,24.5\du){User $K$};
\definecolor{dialinecolor}{rgb}{0.000000, 0.000000, 0.000000}
\pgfsetstrokecolor{dialinecolor}
\node[anchor=west] at (23.685420\du,23.537316\du){};

\end{tikzpicture}

%% file: 34bound.tex
\begin{tikzpicture}[line cap=round,line join=round,x=3.55cm,y=1.7cm,
    spy/.style={%
        draw,green,
        line width=1pt,
        rectangle,inner sep=0pt,
    },
]

    \def\spyviewersize{2.4cm}

    \def\spyonclipreduce{0.8pt}

    \def\spyfactorI{50}
    \coordinate (spy-on 1) at (25/12,1/3);
    \coordinate (spy-in 1) at (2.3,1.5);

    \def\pic{\coordinate (O) at (0,0);
       \draw [ultra thin,step=3,black] (0,0) grid (3,3);
%
   \foreach \x in {0,1/4,3/4,6/4,9/4,3}
   \draw[shift={(\x,0)},color=black,thin] (0pt,1pt) -- (0pt,-1pt)
                                   node[below] {\footnotesize $\x$};
      \foreach \y in {0,1,2,3}
      \draw[shift={(0,\y)},color=black,thin] (1pt,0pt) -- (-1pt,0pt)
                                    node[left] {\footnotesize $\y$};
  \draw[color=black] (6cm,-18pt) node[left] { Cache size $M$};
  \draw[color=black] (-18pt,3.cm) node[left,rotate=90] { Rate $R$};
%
%

   \draw[smooth,blue!70,mark=otimes,samples=1000,domain=0.0:2.2,mark = $\otimes$,line width=2pt]
      {(0,3)--(1/4,2.2)  node[mark size=2.5pt,line width=2pt,label={below:$(\frac{1}{4},\frac{9}{4})$}]{$\pgfuseplotmark{square*}$} (9/4,1/4)node[mark size=2.5pt,line width=2pt,label={above:$(\frac{9}{4},\frac{1}{4})$}]{$\pgfuseplotmark{square*}$}--(3,0)};
      
      
      \draw[smooth,red!60!black,samples=1000,domain=0.0:2.2,mark size=2pt,mark =otimes*,line width=2pt]
      {(3/8,2.05)node[mark size=2.5pt,line width=2pt,label={right:$(\frac{3}{8},2)$}]{$\pgfuseplotmark{square*}$}--(1/4,2.2)};
      

    }

    \pic

\draw [gray!50!white,line width=1pt,fill=white] (0.9,2.2)rectangle (3,3);
\begin{scope}[shift={(0.75,2.4)}] 
\draw [smooth,samples=1000,domain=0.0:2.2,red!60!black,mark=otimes,line width=2pt] 
{(0.25,0) --node [mark size=3pt,line width=0.1pt]{$\pgfuseplotmark{square*}$} (0.5,0)}
node[right]{New Rate  Memory Tradeoff};

\draw [yshift=0.75\baselineskip,smooth,blue!70,samples=1000,domain=0.0:2.2,mark=otimes,line width=2pt] 
{(0.25,0) --node [mark size=3pt,line width=0.1pt]{$\pgfuseplotmark{square*}$} (0.5,0)}
node[right]{Rate Memory  Tradeoff \cite{maddah2014fundamental,chen2014fundamental}};
	\end{scope}

\end{tikzpicture}

%% file: 24bound.tex
\begin{tikzpicture}[line cap=round,line join=round,x=5.45cm,y=2.8cm,
    spy/.style={%
        draw,green,
        line width=1pt,
        rectangle,inner sep=0pt,
    },
]

    \def\spyviewersize{2.4cm}

    \def\spyonclipreduce{0.8pt}

    \def\spyfactorI{50}
    \coordinate (spy-on 1) at (25/12,1/3);
    \coordinate (spy-in 1) at (2.3,1.5);

    \def\pic{\coordinate (O) at (0,0);
       \draw [ultra thin,step=2,black] (0,0) grid (2,2);
%
   \foreach \x in {0,1/4,1/2,1,3/2,2}
   \draw[shift={(\x,0)},color=black,thin] (0pt,1pt) -- (0pt,-1pt)
                                   node[below] {\footnotesize $\x$};
      \foreach \y in {0,1,2}
      \draw[shift={(0,\y)},color=black,thin] (1pt,0pt) -- (-1pt,0pt)
                                    node[left] {\footnotesize $\y$};
  \draw[color=black] (6cm,-18pt) node[left] { Cache size $M$};
  \draw[color=black] (-18pt,3.cm) node[left,rotate=90] { Rate $R$};
%

   \draw[smooth,blue!70,mark=otimes,samples=1000,domain=0.0:2.2,mark = $\otimes$,line width=2pt]
      {(0,2)--(1/4,3/2)  node[mark size=2.5pt,line width=2pt,label={above:$(\frac{1}{4},\frac{3}{2})$}]{$\pgfuseplotmark{square*}$} (3/2,1/4)node[mark size=2.5pt,line width=2pt,label={above:$(\frac{3}{2},\frac{1}{4})$}]{$\pgfuseplotmark{square*}$}--(2,0)};
      \draw[smooth,purple!70!black,mark=otimes,samples=1000,domain=0.0:2.2,mark = $\otimes$,dash pattern=on 6pt off 6pt,line width=2pt]
      {(1/4,3/2)--(1/2,1)  node[mark size=2.5pt,line width=2pt,label={below:$(\frac{1}{2},1)$}]{$\pgfuseplotmark{square*}$} };
      
      \draw[smooth,green!60!black,samples=1000,domain=0.0:2.2,mark size=2pt,mark =triangle*,dash pattern=on 6pt off 6pt,line width=2pt]
      {(1/2,7/6)node[mark size=2pt,line width=2pt,label={right:$(\frac{1}{2},\frac{7}{6})$}]{$\pgfuseplotmark{square*}$}--(1/4,3/2)};
      \draw[smooth,black,samples=1000,domain=0.0:2.2,mark size=2pt,mark =triangle*,dash pattern=on 6pt off 6pt,line width=2pt]
      {(1/2,11/9)node[mark size=2pt,line width=2pt,label={above:$(\frac{1}{2},\frac{11}{9})$}]{$\pgfuseplotmark{square*}$}--(1/4,3/2)};

    }

    \pic

\draw [gray!50!white,line width=1pt,fill=white] (0.6,1.35)rectangle (2,2.15);
\begin{scope}[shift={(0.44,1.45)}] 
		\draw [smooth,samples=1000,domain=0.0:2.2,green!60!black,mark=otimes,dash pattern=on 6pt off 6pt,line width=2pt] 
	{(0.25,0) --node [mark size=2.5pt,line width=0.1pt]{$\pgfuseplotmark{square*}$} (0.5,0)}
	node[right]{New Lower Bound};
	\draw[yshift=0.75\baselineskip,smooth,purple!70!black,mark=otimes,samples=1000,domain=0.0:2.2,mark = $\otimes$,dash pattern=on 6pt off 6pt,line width=2pt]
	{(0.25,0) --node [mark size=2.5pt,line width=0.1pt]{$\pgfuseplotmark{square*}$} (0.5,0)}
	node[right]{Known Lower Bound \cite{maddah2014fundamental,sengupta2015improved}};

		\draw [yshift=2.25\baselineskip,smooth,blue!70,samples=1000,domain=0.0:2.2,mark=otimes,line width=2pt] 
				{(0.25,0) --node [mark size=2.5pt,line width=0.1pt]{$\pgfuseplotmark{square*}$} (0.5,0)}
				node[right]{Rate Memory Tradeoff \cite{maddah2014fundamental,chen2014fundamental}};
		
		\draw [yshift=1.5\baselineskip,smooth,samples=1000,domain=0.0:2.2,black,mark=otimes,dash pattern=on 6pt off 6pt,line width=2pt] 
		{(0.25,0) --node [mark size=2.5pt,line width=0.1pt]{$\pgfuseplotmark{square*}$} (0.5,0)}
		node[right]{Known Achievable  Rate \cite{tian2016symmetry}};
	\end{scope}

\end{tikzpicture}

%% file: NKbound.tex
\begin{tikzpicture}[line cap=round,line join=round,x=3.8cm,y=2cm,
spy/.style={%
	draw,green,
	line width=1pt,
	rectangle,inner sep=0pt,
},
]

\def\spyviewersize{2.4cm}

\def\spyonclipreduce{0.8pt}

\def\spyfactorI{50}
\coordinate (spy-on 1) at (25/12,1/3);
\coordinate (spy-in 1) at (2.3,1.5);

\def\pic{\coordinate (O) at (0,0);
	\draw [ultra thin,step=3,black] (0,0) grid (3,3);
	%
	\foreach \x in {0,1/4,3/4,6/4,9/4,3}
	\draw[shift={(\x,0)},color=black,thin] (0pt,1pt) -- (0pt,-1pt);

	\foreach \y in {0,1,2,3}
	\draw[shift={(0,\y)},color=black,thin] (1pt,0pt) -- (-1pt,0pt);
	\draw[color=black] (6cm,-18pt) node[left] { Cache size $M$};
	\draw[color=black] (-18pt,3.cm) node[left,rotate=90] { Rate $R$};
%
%

	\draw[smooth,blue!60!white,mark=otimes,samples=1000,domain=0.0:2.2,mark = $\otimes$,line width=2pt]
	{(0,3)--(1/4,9/4)  node[mark size=2.5pt,line width=2pt]{$\pgfuseplotmark{square*}$} (9/4,1/4)node[mark size=2.5pt,line width=2pt]{$\pgfuseplotmark{square*}$}--(3,0)};
	\draw[smooth,blue!60!white,mark=otimes,samples=1000,domain=0.0:2.2,mark = $\otimes$,dash pattern=on 2pt off 2pt,line width=1pt]
	{(1/4,9/4)--(1/4,0)  node[mark size=0.5pt,line width=2pt,label={below:$\frac{1}{K}$}]{$\pgfuseplotmark{square*}$} (9/4,0)node[mark size=0.5pt,line width=2pt,label={below:$\frac{N(K-1)}{K}$}]{$\pgfuseplotmark{square*}$}--(9/4,1/4)};
	
	\draw[smooth,red!70!black,samples=1000,domain=0.0:2.2,mark =otimes*,line width=2pt]
	{(0.45,39/20)node[mark size=2.5pt,line width=2pt,]{$\pgfuseplotmark{square*}$}--(1/4,9/4)};
	\draw[smooth,red!70!black,samples=1000,domain=0.0:2.2,mark =otimes*,dash pattern=on 2pt off 2pt,line width=1pt]
	{(0.45,0)node[mark size=0.5pt,line width=2pt,label={below:$\frac{N}{K(N-1)}$}]{$\pgfuseplotmark{square*}$}--(0.45,39/20)};
	
	

}

\pic

\draw [gray!50!white,line width=1pt,fill=white] (0.9,2.2)rectangle (3,3);
\begin{scope}[shift={(0.7,2.5)}] 
\draw [smooth,samples=1000,domain=0.0:2.2,red!60!black,mark=otimes,line width=2pt] 
{(0.25,0) --node [mark size=3pt,line width=0.1pt]{$\pgfuseplotmark{square*}$} (0.5,0)}
node[right]{New Rate Memory Tradeoff};

\draw [yshift=0.75\baselineskip,smooth,samples=1000,domain=0.0:2.2,blue!70,mark=otimes,line width=2pt] 
{(0.25,0) --node [mark size=3pt,line width=0.1pt]{$\pgfuseplotmark{square*}$} (0.5,0)}
node[right]{Rate Memory Tradeoff \cite{maddah2014fundamental,chen2014fundamental}};
\end{scope}

\end{tikzpicture}